\documentclass[12pt,reqno]{amsart}
\usepackage{amsmath,amssymb,amsfonts,amsthm}

\usepackage{amsmath, euscript, amssymb, amsfonts}
\usepackage{mathrsfs}

\theoremstyle{plain}
\newtheorem{theorem}{Theorem}
\newtheorem{lemma}[theorem]{Lemma}
\newtheorem{corollary}[theorem]{Corollary}
\newtheorem{proposition}[theorem]{Proposition}

\theoremstyle{remark}
\newtheorem{remark}[theorem]{Remark}

\numberwithin{equation}{section}
\numberwithin{theorem}{section}

\newcommand{\mS}{{\mathscr{S}}}
\newcommand{\mE}{{\mathscr{E}}}

\newcommand{\mM}{{\mathscr{M}}}

\newcommand{\Op}{\mathop{\mathrm{Op}}\nolimits}

\newcommand{\coloneq}{\mathrel{\mathop:}=}

\newcommand{\oR}{{\mathbb R}}

\newcommand{\oZ}{{\mathbb Z}}
\newcommand{\oN}{{\mathbb N}}

\newcommand{\Sab}{{S^{\{b\}}_{\{a\}}}}
\newcommand{\Sba}{{S_{\{b\}}^{\{a\}}}}
\newcommand{\sab}{{S^{(b)}_{(a)}}}
\newcommand{\sba}{{S_{(b)}^{(a)}}}

\newcommand{\dSab}{{S^{\{b\}\prime}_{\{a\}}}}
\newcommand{\dSba}{{S^{\{a\}\prime}_{\{b\}}}}
\newcommand{\dsab}{{S^{(b)\prime}_{(a)}}}
\newcommand{\dsba}{{S^{(a)\prime}_{(b)}}}

\newcommand{\Eab}{{\mathcal E^{\{b\}}_{(a)}}}
\newcommand{\Eabp}{{\mathcal E^{(b)}_{\{a\}}}}
\newcommand{\Eba}{{\mathcal E^{\{a\}}_{(b)}}}
\newcommand{\Ebap}{{\mathcal E^{(a)}_{\{b\}}}}
\newcommand{\dEba}{{\mathcal E^{\{a\}\prime}_{(b)}}}
\newcommand{\dEbap}{{\mathcal E^{(a)\prime}_{\{b\}}}}

\newcommand{\eab}{{\breve{\mathcal E}^{\{b\}}_{(a)}}}
\newcommand{\eabp}{{\breve{\mathcal E}^{(b)}_{\{a\}}}}
\newcommand{\deab}{{\breve{\mathcal E}^{\{b\}\prime}_{(a)}}}
\newcommand{\deabp}{{\breve{\mathcal E}^{(b)\prime}_{\{a\}}}}
\newcommand{\deba}{{\breve{\mathcal E}^{\{a\}\prime}_{(b)}}}
\newcommand{\debap}{{\breve{\mathcal E}^{(a)\prime}_{\{b\}}}}

\begin{document}

\baselineskip=20pt

\title[Inclusion theorems for the Moyal multiplier algebras]{
 Inclusion theorems for the Moyal multiplier algebras of  generalized Gelfand-Shilov spaces}

\author{M.~A.~Soloviev}
\address{I.~E.~Tamm Department of Theoretical Physics, Lebedev
Physical Institute of the Russian Academy of Sciences,
 Leninsky Prospect 53, Moscow 119991, Russia} \email{soloviev@lpi.ru}
\thanks{}
\subjclass[2010]{53D55, 43A22, 35S05, 46F05}
\keywords{deformation quantization, Weyl symbols, Weyl-Moyal product, multiplier algebras, Gelfand-Shilov spaces, pseudodifferential operators}

\begin{abstract}
 We prove that the Moyal multiplier algebras of the generalized  Gelfand-Shilov spaces of type $S$  contain Palamodov spaces of type $\mathscr E$ and the  inclusion maps are continuous.  We also give a direct proof that the Palamodov spaces are algebraically and topologically isomorphic to the strong duals of the spaces of convolutors for the corresponding spaces of type $S$. The obtained results  provide an  effective way to describe the properties of pseudodifferential operators with symbols in the spaces of type $\mathscr E$.
\end{abstract}

\maketitle
\section{Introduction}
\label{S1}

The   Gelfand-Shilov spaces of type $S$  provide a natural framework for studying  infinite order pseudodifferential operators whose symbols   have faster than polynomial
growth at infinity. Various classes of pseudodifferential operators of this kind were investigated in recent papers~\cite{ACT,CPP,CT,P}
with special attention to their continuity, composition, and invariance properties.
Similar issues arise in the context of noncommutative quantum field theory, where spaces of type $S$ were used to characterize the violations of locality and causality~\cite{Ch,S2007-II} and to analyze the behavior of propagators in some noncommutative models~\cite{FS,FS11,Zahn}. It is crucial for these applications that, under natural restrictions specified in~\cite{S2007},
 the Gelfand-Shilov spaces of functions  on the linear symplectic space $\oR^{2d}$ are algebras under the Weyl-Moyal product
\begin{equation}
(f\star_\hbar g)(x)=
(\pi\hbar)^{-2d}\int_{\oR^{4d}} f(x-x')g(x-x^{\prime\prime})e^{(2i/\hbar)x'\cdot
Jx^{\prime\prime}}dx'dx^{\prime\prime},
 \label{1.1}
\end{equation}
where  $J=\begin{pmatrix} 0&I_d\\-I_d&0\end{pmatrix}$  is the standard symplectic matrix,
 $x=(q,p)$,  $x'\!\cdot\! Jx''= q'p''-q'' p'$, and $\hbar$ is the Planck constant.
In this and previous papers~\cite{S2019,S2019-2,S2020}, we study the   algebras of multipliers of   generalized Gelfand-Shilov spaces with respect to the noncommutative product~\eqref{1.1}.
 Their importance is due to the fact that these algebras extend this operation to the maximum possible class of functions, including some elements of the  duals of spaces of type $S$.
  From the viewpoint of the Weyl symbol calculus, the multiplier  algebras consist of the symbols of the operators that map the corresponding Fourier-invariant spaces of type $S$ continuously into itself. These algebras generalize the Moyal multiplier algebra
  $\mathscr M_\hbar(S)$ for the Schwartz space $S(\oR^{2d})$ of all infinitely differentiable rapidly decreasing functions.
  The algebra  $\mathscr M_\hbar(S)$ has been studied in many papers starting from~\cite{A1,G-BV1,Mail,G-BV2} and its applications to quantum field theory on noncommutative spaces have been discussed in~\cite{G2004,G-BV3}.  We treat the product~\eqref{1.1} as a deformation of  the ordinary pointwise product and use the notation  $\star_\hbar$, accepted in  deformation quantization theory,  instead of the notation $\#$ which is used for the composition  of  Weyl symbols in~\cite{Fol,G2001,H3} and corresponds to $\hbar=1$.

  Typically, applications use Gelfand-Shilov spaces of a particular type, denoted in~\cite{GS2} by $S^\beta_\alpha$. The algebras of their poitnwise multipliers have been explicitly described by Palamodov~\cite{P1962} in terms of  spaces of type  $\mathscr E$. The noncommutative deformation violates the equalities established in~\cite{P1962}, but preserves some inclusion relations which are  the subject of our study.  The starting point for us is that the Moyal multiplier  algebras of the spaces of type $S$ contain the duals of their convolutor spaces.
   This inclusion was proved  for   $S^\beta_\alpha$ in~\cite{S2011} and  holds true in the general case, as shown in~\cite{S2019, S2020}.
  The spaces of convolutors for Gelfand-Shilov spaces were studied in~\cite{DPV,DPPV} and, most thoroughly, by Debrouwere and Vindas in~\cite{DV2018,DV2019}, where the short-time Fourier transform and a projective description of inductive limits were used for this purpose.    Here we develop
   an alternative  approach based on the  continuous extension of the convolutors of  spaces of type $S$ to the corresponding spaces of type  $\mathscr E$. A simple and natural way of such extension was  proposed earlier in~\cite{S2012-I} and was also used  in~\cite{S2020}.  We show that this approach provides a complete characterization of the convolutor spaces for the generalized spaces of type $S$. Significantly, it is also well suited for the cases that have not yet been considered,
   for example, where a space of type $S$ is nontrivial, whereas its projective counterpart is trivial.     In combination with Theorem~4 of~\cite{S2019-2}, this approach provides a direct and simple proof of the continuous embedding of  generalized Palamodov spaces of type $\mathscr E$ into the Moyal multipliers algebras of the corresponding spaces of type $S$, which is the main result of this paper.

 The paper is structured as follows. In Section~\ref{S2}, we give the basic definitions concerning spaces of type $S$ and  type $\mathscr E$ and  introduce the notation.
 We try to follow the original notation in~\cite{GS2,P1962} and let  $a_n$  and $b_n$ denote the sequences defining the function spaces and specifying the behavior at infinity and the degree of smoothness of their elements. But instead of the notation $S^{b_n}_{a_k}$ introduced in~\cite{GS2}, we use $S^{\{b\}}_{\{a\}}$, where the curly brackets mean that this space is the inductive limit of a family of normed spaces. The projective limit of the same family is denoted by  $S^{(b)}_{(a)}$. Such a rule was used  by Komatsu~\cite{K1973} and  in many subsequent papers, and we apply it also to the spaces of type~$\mathscr E$.  In Section~\ref{S3}, we prove that the convolutor spaces for  $S^{\{b\}}_{\{a\}}$ and  $S^{(b)}_{(a)}$  contain respectively the duals of the Palamodov spaces  $\eab$ and $\eabp$ and that these inclusions are continuous. In Section~\ref{S4}, we give a complete characterization of the convolutor spaces  $C\bigl(S^{\{b\}}_{\{a\}}\bigr)$ and $C\bigl(S^{(b)}_{(a)}\bigr)$ in terms of the spaces of type $\mathscr E$.
 In particular, we show that  $\eab$ is canonically isomorphic to the strong dual of  $C\bigl(S^{\{b\}}_{\{a\}}\bigr)$ and  $\eabp$ is canonically isomorphic to the strong dual of  $C\bigl(S^{(b)}_{(a)}\bigr)$. In Section~\ref{S5}, we define the left, right, and  two-sided multipliers for the noncommutative algebras $\bigl(\Sab,\star_\hbar\bigr)$ and $\bigl(\sab,\star_\hbar\bigr)$ and prove that these algebras have approximate identities. This allows us to prove the equivalence of  three different definitions of their corresponding multiplier algebras. In Section~\ref{S6}, we show that  $\eab$ is continuously embedded in the algebra   ${\mathscr M}_\hbar\bigl(\Sab\bigr)$  of two-sided Moyal multipliers for $\Sab$ and $\eabp$ is continuously embedded in ${\mathscr M}_\hbar\bigl(\sab\bigr)$. In the same section we extend these theorems   to other translation invariant star products. Additional inclusion relations are established in the case of the Fourier-invariant spaces $S_{\{a\}}^{\{a\}}$ and $S_{(a)}^{(a)}$. Section~\ref{S7} contains concluding remarks.

\section{Preliminaries and notation}
\label{S2}
    Let $a=(a_n)_{n\in\oZ_+}$    be a sequence  of positive numbers such that
\begin{equation}
a_0=1,\qquad a_{n+1}\ge a_n,
 \label{2.1}
 \end{equation}
 \begin{equation}
a_n^2\le a_{n-1} a_{n+1},
 \label{2.2}
 \end{equation}
  \begin{equation}
 a_{k+n}\le K H^{k+n}a_ka_n,
 \label{2.3}
 \end{equation}
where $K$ and  $H$ are  positive constants.  The logarithmic convexity condition~\eqref{2.2} coupled with the normalization condition $a_0=1$ implies the inequality
\begin{equation}
a_ka_n\le a_{k+n},
 \label{2.4}
 \end{equation}
 which will be also  used throughout the paper. Let $b=(b_n)_{n\in\oZ_+}$
     be another sequence  of positive numbers satisfying the same conditions.
 The Gelfand-Shilov space $\Sab(\oR^d)$ consists of all infinitely differentiable functions $f(x)$ defined on $\oR^d$ and satisfying the inequalities
 \begin{equation}
 |x^\alpha\partial^\beta f(x)|\le C A^{|\alpha|}B^{|\beta|} a_{|\alpha|}b_{|\beta|}\qquad \forall \alpha,\beta\in \oZ_+^d,
 \notag
 \end{equation}
   where $C$, $A$, and  $B$ are positive constants depending on $f$,
  $\oZ^d_+$  is  the set of $d$-tuples of nonnegative integers,
  and the standard multi-index notation is used.  In what follows, we write for brevity  $\Sab$ instead of $\Sab(\oR^d)$ when this cannot cause confusion. This space is the union of a family of Banach spaces $\{S^{b,B}_{a,A}\}_{A,B>0}$ whose norms are given by
 \begin{equation}
 \|f\|_{A,B}= \sup_{x,\alpha,\beta} \frac{|x^\alpha\partial^\beta f(x)|}{A^{|\alpha|}B^{|\beta|} a_{|\alpha|}b_{|\beta|}},
 \label{2.5}
 \end{equation}
 and its topology is defined to be the inductive limit topology with respect to the
 inclusion maps   $S^{b,B}_{a,A}\to \Sab$. The most frequently used
Gelfand-Shilov spaces $S^\beta_\alpha$ are defined in~\cite{GS2} by sequences of the form
\begin{equation}
a_n=n^{\alpha n},\quad b_n=n^{\beta n},
\label{2.6}
 \end{equation}
  where in this case $\alpha$ and $\beta$ are nonnegative numbers, which should not be confused with the multi-indices in~\eqref{2.5}.     We will also consider the spaces
\begin{equation}
\sab=\bigcap_{A\to0,B\to0} S^{b,B}_{a,A}
\label{2.7}
 \end{equation}
 equipped with the projective limit topology.
  If $\varliminf_{n\to\infty}a_n^{1/n}=0$, then the spaces $\Sab$ and $\sab$   are trivial, i.e., contain only the identically zero function. If $0<\varliminf_{n\to\infty}a_n^{1/n}<\infty$, then all functions in $\Sab$ are of compact support, and this space coincides with the space defined by $b_n$
  and  $a_n\equiv 1$, which is usually denoted by  $\mathcal D^{\{b\}}$.  The space   $\sab$ is trivial in this case, and  considering the spaces~\eqref{2.7}  we always assume that
 \begin{equation}
\lim_{n\to\infty}a_n^{1/n}=\infty.
\label{2.8}
 \end{equation}
 There are also other non-triviality conditions for the spaces of type $S$, see~\cite{GS2}. Their precise description is not needed for what follows, but we assume throughout the paper that the spaces under consideration are nontrivial. A non-quasianalyticity condition is often imposed on
  $b_n$ to   ensure that the space contains  sufficiently many functions of compact support (see~\cite{GS2,K1973}), but this condition is not used in the proofs given below.
 The Fourier transformation
\begin{equation}
F\colon f(x)\to
\widehat f(\zeta)=(2\pi)^{-d/2}\int_{\oR^d}e^{-ix\cdot\zeta}f(x)dx
 \notag
 \end{equation}
 maps $\Sab$ isomorphically onto $\Sba$ and maps $\sab$ isomorphically onto $\sba$. Using~\eqref{2.4}, it is easy to see that $\Sab$ and $\sab$ are   algebras under pointwise multiplication and that this operation is  continuous in their topologies. As a consequence, these spaces are also topological algebras under convolution.
The norm~\eqref{2.5} can be written as
 \begin{equation}
 \|f\|_{A,B}= \sup_{x,\beta}\frac{w_a(|x|/A)|\partial^\beta f(x)|}{B^{|\beta|} b_{|\beta|}},
 \label{2.9}
 \end{equation}
 where $|x|=\max\limits_{i\le j\le d}|x_j|$ and
 \begin{equation}
 w_a(t)\coloneq\sup_{n\in\oZ_+}\frac{t^n}{a_n},\qquad t\ge 0.
\label{2.10}
 \end{equation}
 This function is often called a weight function. We note in this connection that the replacement  $\sup_x|\cdot|\to\|\cdot\|_{L^1}$ in~\eqref{2.9} gives an equivalent system of norms  (see, e.g,  Lemma A.2 in~\cite{S2019} for a proof) and then $w_a(t)$ plays the role of a weight in  the integral.
  Under  condition~\eqref{2.8},  the function $w_a(t)$ is finite and continuous. If $a_n\equiv 1$, then its corresponding function $w_1(t)$ is equal to  1 for $0\le t\le1$  and is infinite for $t>1$, i.e.,
  $1/w_1(t)$ is the characteristic function of the interval  $[0,1]$. It follows from the definition and from~\eqref{2.1} that $w_a(t)\ge 1$ and  that this function is convex and monotonically increases faster than  $t^n$ for any $n$. Therefore,
  \begin{equation}
 w_a\left(\frac{t_1+ t_2}{A_1+A_2}\right)\le  w_a\left(\frac{t_1}{A_1}\right)\, w_a\left(\frac{t_2}{A_2}\right)
\label{2.11}
 \end{equation}
  for any positive $A_1$ and $A_2$. In particular, $w_a((t_1+ t_2)/2)\le  w_a(t_1)\, w_a(t_2)$.  Setting $t_1=|x|$ and $t_2=|y|$, we obtain the inequality
   \begin{equation}
 w_a\left(\frac12|x+y|\right)\le  w_a(|x|)\, w_a(|y|),
\label{2.12}
 \end{equation}
 which is most often used below.  The condition~\eqref{2.3} implies that
  \begin{equation}
 w_a(t)^2\le K w_a(Ht).
\label{2.13}
 \end{equation}

 Along with the spaces of rapidly decreasing functions, we will consider spaces of rapidly increasing functions with the same degree of smoothness. Namely,
 let $\mathcal E^{b,B}_{a,A}$ be the Banach space of all functions  with the finite norm\footnote{In the case where  $a_n\equiv 1$, the functions  in $\mathcal E^{b,B}_{a,A}$  are defined only for  $|x|\le A$ and are regarded as zero if they are zero in this domain.}
  \begin{equation}
 \|h\|^A_B= \sup_{x,\beta} \frac{|\partial^\beta h(x)|}{w_a(|x|/A)\,B^{|\beta|} b_{|\beta|}}.
 \label{2.14}
 \end{equation}
 Using the family $\{\mathcal E^{b,B}_{a,A}\}_{A,B>0}$,   we define the spaces
 \begin{gather}
 \Eab\coloneq \bigcap_{A\to\infty}\bigcup_{B\to\infty}\mathcal E^{b,B}_{a,A},\qquad
 \Eabp\coloneq \bigcap_{B\to0}\bigcup_{A\to0}\mathcal E^{b,B}_{a,A},
 \label{2.15}\\
 \eab\coloneq \bigcup_{B\to\infty}\bigcap_{A\to\infty}\mathcal  E^{b,B}_{a,A},\qquad
 \eabp\coloneq \bigcup_{A\to0}\bigcap_{B\to0}\mathcal  E^{b,B}_{a,A},
 \label{2.16}
 \end{gather}
 where the intersections are endowed with the projective limit topology and the unions are endowed with the inductive limit topology.
 The spaces~\eqref{2.15} and~\eqref{2.16} play the same role in the theory of ultradistributions defined on  $\Sab$ and $\sab$  as the spaces $\mathcal O_M$ and $\mathcal O_C$ in the Schwartz theory of tempered distributions~\cite{Schwartz}. They were introduced by Palamodov~\cite{P1962} (with somewhat different notation) for the case~\eqref{2.6}  and were called spaces of type   $\mE$. This terminology is quite natural because the notation   $\mE(\oR^d)$ was used by Schwartz for the space of all infinitely differentiable functions on $\oR^d$.
 In the case where $a_n\equiv 1$, we write $\mathcal E^{\{b\}}$ instead of  $\mathcal E^{\{b\}}_{(1)}$ and $\breve{\mathcal E}^{\{b\}}$ instead of  $\breve{\mathcal E}^{\{b\}}_{(1)}$, which is in agreement with  the notation in~\cite{K1973}.
  Obviously, we have the  continuous inclusions
 \begin{equation}
\eab\hookrightarrow \Eab,\qquad  \eabp\hookrightarrow \Eabp.
\label{2.17}
 \end{equation}
 It should be noted  that the spaces $\mathcal O_C^{\{M_p\},\{A_p\}}(\oR^d)$ and $\mathcal O_C^{(M_p), (A_p)}(\oR^d)$ considered in~\cite{DV2019}  are respectively $\breve{\mathcal E}^{\{M\}}_{(A)}$ and  $\breve{\mathcal E}^{(M)}_{\{A\}}$ in our notation, and the symbol classes $\Gamma_s^\infty(\mathbf R^d)$ and $\Gamma_{0,s}^\infty(\mathbf R^d)$ studied in~\cite{CT}  coincide respectively with $\breve{\mathcal E}^{\{a\}}_{(a)}$ and  $\breve{\mathcal E}^{(a)}_{\{a\}}$, where $a_n=n^{sn}$. The spaces  $S_{\{a\}}^{\{a\}}$ and $S_{(a)}^{(a)}$ with $a_n=n^{sn}$ are  denoted in~\cite{CT} by $\mathcal S_s(\mathbf R^d)$ and $\Sigma_s(\mathbf R^d)$.   We also note that the notation  in~\cite{S2019,S2019-2,S2020} is related to the notation used here as follows: $S^b_a\leftrightarrow \Sab$, $\mS^b_a\leftrightarrow \sab$,  $E^b_a \leftrightarrow\Eab$, $\mE^b_a\leftrightarrow\Eabp$,  $\breve E^b_a \leftrightarrow\eab$, $\breve \mE^b_a\leftrightarrow\eabp$, $a(t)\leftrightarrow w_a(t)$.

 The spaces of type  $S$ and type $\mE$ have nice topological properties which follow from the fact that the  inclusion maps $S^{b,B}_{a,A}\hookrightarrow S^{b,\bar B}_{a,\bar A}$ and   $\mathcal E^{b,B}_{a,\bar A}\hookrightarrow \mathcal E^{b,\bar B}_{a, A}$, where $A<\bar A$ and  $B<\bar B$, are compact. This is well known for the spaces of type   $S$ and is proved in the same manner as in~\cite{GS2}  for $S^\beta_\alpha$. In the case of spaces of type  $\mE$,  the compactness of the inclusion maps can also be proved analogously, see, e.g., Lemma~2 in~\cite{S2019-2}. It follows that
  $\sab$  is an  (FS)-space (Fr\'echet-Schwartz space) and   $\Sab$ belongs to the dual class of (DFS)-spaces (see~\cite{K1967,MV1997} for the basic properties of (FS) and (DFS)-spaces). In particular, these space are complete, barrelled, reflexive, and  Montel. An important consequence is that the inductive limit   $\Sab$ is regular, i.e., every bounded subset of it is contained and bounded in some  $S^{b,B}_{a,A}$. We let $\dSab$ and $\dsab$ denote the dual spaces of
  $\Sab$ and  $\sab$ and assume that the duals are  endowed with the strong topology. Then the first of them is  an (FS)-space and the second is a (DFS)-space.  The inductive system of spaces $\mathcal E_{(a)}^{b,B}=\bigcap_{A\to\infty}\mathcal E^{b,B}_{a,A}$,  $B>0$, is equivalent to the system $\mathcal E_{(a)}^{b,B+}=\bigcap_{A\to\infty,\epsilon\to0}\mathcal E^{b,B+\epsilon}_{a,A}$, because there are continuous inclusions $\mathcal E_{(a)}^{b,B}\subset\mathcal E_{(a)}^{b,B+}\subset \mathcal E_{(a)}^{b,\bar B}$ for $B<\bar  B$. In turn, the inductive system $\mathcal E_{a,A}^{(b)}=\bigcap_{B\to0}\mathcal E^{b,B}_{a,A}$, $A>0$, is equivalent to the system $\mathcal E_{a,A-}^{(b)}=\bigcap_{B\to0, \epsilon\to0} E^{b,B}_{a,A-\epsilon}$.
As a consequence, the projective system of the duals   $\bigl(\mathcal E_{(a)}^{b,B}\bigr)'$ is equivalent to the system $\bigl(\mathcal E_{(a)}^{b,B+}\bigr)'$ and the projective system $\bigl(\mathcal E_{a,A}^{(b)}\bigr)'$ is equivalent to the system $\bigl(\mathcal E_{a,A-}^{(b)}\bigr)'$.
 By Lemma~2 in~\cite{S2019-2},  $\mathcal E_{(a)}^{b,B+}$ and $\mathcal E_{a,A-}^{(b)}$ are (FS)-spaces and their strong duals are hence (DFS)-spaces. The representations
 \begin{equation}
 \eab=\varinjlim_{B\to\infty}\mathcal E_{(a)}^{b,B+} ,\qquad  \eabp=\varinjlim_{A\to0}\mathcal E_{a,A-}^{(b)}
\label{2.18}
 \end{equation}
 play an essential role in our study.

For any topological vector space $E$, we let $\mathcal L(E)$ denote the space of all continuous linear maps  $E\to E$, equipped with the topology of uniform convergence on bounded subsets.   If $E$  is an   (FS)-space or a  (DFS)-space, then  $\mathcal L(E)$ is complete (see  Sect.~39.6 in~\cite{K1979}). If $E$ is a test-function space on  $\oR^d$,  then a functional $u\in E'$  is called a convolutor for $E$ if the function
 \begin{equation}
 (u\ast f)(x)= \langle u,f(x-\cdot)\rangle
 \notag
 \end{equation}
    belongs to $E$ for any $f\in E$ and the map  $f\to u\ast f$ is continuous on $E$.  In the case of the spaces of type $S$,
 the continuity condition is automatically satisfied by the closed graph theorem
   (Theorem~8.5 in Ch.~IV in~\cite{Sch}).  The set of all convolutors for  $E$  is denoted by $C(E)$ and is equipped with the topology induced by that of $\mathcal L(E)$. The Fourier transformation maps  $C\bigl(\Sab\bigr)$ and $C\bigl(\sab\bigr)$ isomorphically onto the spaces of pointwise multipliers for   $\Sba$ and $\sba$, which  we respectively denote by $M\bigl(\Sba\bigr)$ and  $M\bigl(\sba\bigr)$.

\section{The duals of  spaces of type $\mathscr E$ as spaces of convolutors  }
\label{S3}

Clearly, we have the canonical continuous inclusions
\begin{equation}
 \Sab\hookrightarrow  \eab, \qquad  \sab\hookrightarrow  \eabp.
 \label{3.1}
 \end{equation}

\begin{proposition}
\label{P3.1}
  If the spaces $\Sab$ and $\sab$ are nontrivial, then the inclusion maps~\eqref{3.1} have dense ranges.
\end{proposition}
 This was proved in~\cite{S2019} and we reproduce the proof below in the course of proving Theorems~\ref{T4.1} and~\ref{T4.3}.
 It follows that the adjoint maps  $\deab\to \dSab$  and $\deabp\to \dsab$  are injective.
  Hence $\deab$  and  $\deabp$ are naturally identified with  vector subspaces of $\dSab$ and  $\dsab$, respectively.

\begin{proposition}
\label{P3.2}
  The space $\deab$  is contained in the convolutor space   $C\bigl(\Sab\bigr)$ and   $\deabp$ is contained in $C\bigl(\sab\bigr)$.
  \end{proposition}

\begin{proof}  Standard arguments show that for any  $u\in  \dSab$ and $f\in \Sab$, the convolution  $(u\ast f)(x)$ is infinitely differentiable and $\partial^\beta(u\ast f)=u\ast\partial^\beta f$ (see Lemma~A.4 in~\cite{S2019}). If  $u$ belongs to $\deab$, then it is continuous on every space $\mathcal E^{b,B}_{(a)}=\varprojlim_{A\to\infty}\mathcal E^{b,B}_{a,A}$, $B>0$, and is bounded in some norm  $\|\cdot\|^A_B$, where $A$ depends on $B$. Hence,
\begin{equation}
 |\partial^\beta(u\ast f)(x)|\le \|u\|^A_B\, \|\partial^\beta f(x-\cdot)\|^A_B.
 \label{3.2}
 \end{equation}
 Let $f\in S^{b,B_0}_{a,A_0}$ and  $B\ge HB_0$. Using~\eqref{2.3} applied to $b_n$  and~\eqref{2.11}, we obtain
\begin{multline}
 \|\partial^\beta f(x-\cdot)\|^A_B=\sup_{y,\gamma}
 \frac{|\partial^{\beta+\gamma}f(x-y)|}{w_a(|y|/A)B^{|\gamma|}b_{|\gamma|}}\\
 \le \|f\|_{A_0,B_0} \sup_{y,\gamma}
 \frac{B_0^{|\beta+\gamma|}b_{|\beta+\gamma|}}{w_a(|y|/A)w_a(|x-y|/A_0)B^{|\gamma|}b_{|\gamma|}}
  \le K \|f\|_{A_0,B_0} \frac{(H B_0)^{|\beta|}b_{|\beta|}}
 {w_a(|x|/(A+A_0))}.
 \label{3.3}
 \end{multline}
It follows from~\eqref{3.2} and~\eqref{3.3} that $u\ast f\in S^{b,HB_0}_{a,A+A_0}$ and
\begin{equation}
 \|u\ast f\|_{A+A_0,HB_0}\le K\|u\|^A_B\, \|f\|_{A_0,B_0}.
 \label{3.4}
 \end{equation}
Hence, $\deab\subset C\bigl(\Sab\bigr)$. The proof holds true for $a_n\equiv 1$. In this case, the supports of the functions $1/w_1(|y|/A)$ and  $1/w_1(|x-y|/A_0)$ are disjoint  for $|x|> A+A_0$, and   $(u\ast f)(x)=0$  for each $x$ in this region.
If  $u\in\deabp$, then $u$ is continuous on every space $\mathcal E_{a,A}^{(b)}=\varprojlim_{B\to0}\mathcal E^{b,B}_{a,A}$, $A>0$, and is bounded in some norm $\|\cdot\|^A_B$, where $B$ depends on $A$. For $f\in\sab$, the norm $\|f\|_{A_0,B_0}$ is finite for arbitrarily small $A_0$ and $B_0$. Choosing $B_0\le B/H$, we again arrive at~\eqref{3.4}, which implies the inclusion $\deabp\subset C\bigl(\sab\bigr)$.
\end{proof}

\begin{proposition}
\label{P3.3}  The inclusion maps
 \begin{equation}
\deab\hookrightarrow C\bigl(\Sab\bigr),\qquad  \deabp\hookrightarrow C\bigl(\sab\bigr)
\label{3.5}
 \end{equation}
 are continuous.
\end{proposition}

\begin{proof} By the duality between projective and inductive topologies
(Sect.~IV.4.5 in~\cite{Sch}), it follows from~\eqref{2.18} that
  $\deab$ is algebraically identified with
$\varprojlim_{B\to \infty}\bigl(\mathcal E^{b,B+}_{(a)}\bigr)'$. A base of neighborhoods of the origin in $\deab$ is formed by the polars of the bounded subsets of $\eab$, and this topology is finer than the projective limit topology because every  bounded subset of $\mathcal E^{(b,B)}_{(a)}$ is bounded in $\eab$. We show that the projective limit topology is  in turn finer than the topology induced by $C\bigl(\Sab\bigr)$.
 Let $Q$ be a bounded set in $\Sab$, let $V$ be  a $0$-neighborhood  in $\Sab$, and let  $W_{Q,V}$ be the set of the operators in $\mathcal L\bigl(\Sab\bigr)$ that map $Q$ into $V$. The family of all sets $W_{Q,V}$ with various $Q$ and $V$ forms a base of $0$-neighborhoods  in $\mathcal L\bigl(\Sab\bigr)$. We assert that  for every $W_{Q,V}$, there exists a $0$-neighborhood $U$ in  $\varprojlim_{B\to \infty}\bigl(\mathcal E^{b,B+}_{(a)}\bigr)'$ such that all operators of convolution
by elements of  $U$ are contained in $W_{Q,V}$. Taking the projective  limit
$\mathcal E^{b,B+}_{(a)}=\varprojlim_{A\to\infty,\epsilon\to0} \mathcal E^{b,B+\epsilon}_{a,A}$ in the reduced form, i.e., replacing every space $\mathcal E^{b,B+\epsilon}_{a,A}$, $B>0$, with the closure  of  $\mathcal E^{b,B+}_{(a)}$ in this space and letting $E^{b,B+\epsilon}_{a,A}$ denote this closure, we have $\bigl(\mathcal E^{b,B+}_{(a)}\bigr)'=\varinjlim_{A\to\infty,\epsilon\to0}\bigl(E^{b,B+\epsilon}_{a,A}\bigr)'$. By Theorem~11 in ~\cite{K1967}, this identity holds algebraically and topologically because the inclusion maps $E^{b,B}_{a,\bar A}\hookrightarrow  E^{b,\bar B}_{a, A}$, where $A<\bar A$ and  $B<\bar B$, are compact.
To simplify the notation, we use an injective sequence of spaces
that is equivalent to the  system
  $S^{b,B}_{a,A}$. Since the inductive limit $\Sab$ is regular, the set $Q$ is contained in  $S_{a,k}^{b,k/H}$ with sufficiently large $k\in\oN$ and is bounded in its norm by a constant  $C$. By the definition of the inductive topology, $V$ contains a set of the form $\bigcup_{n\in \oN}\sum_{m\le n}V_m$, where $V_m=\{f\in S_{a,m}^{b,m}\colon \|f\|_{m,m}\le\varepsilon_m\}$ with some $\varepsilon_m>0$. Let $U$ be the intersection of $\deab$ with a $0$-neighborhood in $\bigl(\mathcal E^{b,k+}_{(a)}\bigr)'=\varinjlim_{m\to\infty}\bigl(E^{b,k+1/m}_{a,m}\bigr)'$ which has the form
     $\bigcup_{n\in \oN}\sum_{m\le n}U_m$, where $U_m=\{u\in \bigl(\mathcal E^{b,k+}_{(a)}\bigr)'\colon \|u\|^{m}_{k+1/m}\le\varepsilon_{k+m}/(KC) \}$.
    It follows from~\eqref{3.4} that $u\ast f\in S^{b,k}_{a,k+m}$ and $\|u\ast f\|_{k+m,k}\le\varepsilon_{k+m}$ for all $f\in Q$ and $u\in U_m$. Because
 $\|u\ast f\|_{k+m,k+m}\le \|u\ast f\|_{k+m,k}$, we conclude that the operator of convolution by $u\in U_m$ maps $Q$ into  $V_{k+m}$. Therefore, all the operators of convolution by elements of $U$ map  $Q$ into  $V$, as claimed.

  We  now show  that the inclusion map $\varprojlim_{A\to 0}\bigl(\mathcal E_{a,A-}^{(b)}\bigr)'\to C\bigl(\sab\bigr)$ is continuous.
  Each bounded set
  $Q\subset\sab$ is bounded in the norm of every space $S_{a,1/m}^{b,1/(Hm)}$ by a constant $C_m$.
  Any $0$-neighborhood $V$ in  $\sab$ contains the intersection of $\sab$ with a set of the form $V_{k,\varepsilon}=\{f\in S_{a,1/k}^{b,1/k}\colon \|f\|_{1/k,1/k}\le \varepsilon\}$, where $k\in\oN$ and $\varepsilon>0$.
   We take $U$ to be the intersection of $\deabp$ with the $0$-neighborhood in   $\bigl(\mathcal E^{(b)}_{a,1/k-}\bigr)'$,
  which is the  absolutely convex hull of the union of the sets  $U_m=\{u\in \bigl(\mathcal E^{(b)}_{a,1/k-}\bigr)'\colon \|u\|^{1/k-1/m}_{1/m}\le\varepsilon/(KC_m) \}$,  $m> k$. It follows from~\eqref{3.4} that
  $\|u\ast f\|_{1/k, 1/m}\le\varepsilon$ for all $f\in Q$ and $u\in U_m$.   Because $\|u\ast f\|_{1/k,1/k}\le \|u\ast f\|_{1/k,1/m}$ for  $m> k$, we conclude that the operators of convolution by elements of $U_m$ map $Q$ into  $V_{k,\varepsilon}$. The set $V_{k,\varepsilon}$ is absolutely convex, and all the operators of convolution by  elements of $U$ hence belong to $W_{Q,V}$. This completes the proof.
  \end{proof}

\section{Complete characterization of the convolutor spaces for the spaces of type $S$}
\label{S4}
Theorem~1 in~\cite{S2020} establishes that every functional  $u\in C\bigl(\Sab\bigr)$ has a unique continuous extension to $\eab$ and every  $u\in C\bigl(\sab\bigr)$ extends uniquely to a continuous linear functional on $\eabp$.  Therefore, there are  natural inclusion maps   $C\bigl(\Sab\bigr)\to\deab$ and $C\bigl(\sab\bigr)\to\deabp$. Clearly, their compositions with the respective inclusions~\eqref{3.5} are
the identity  on $C\bigl(\Sab\bigr)$  and on $C\bigl(\sab\bigr)$. Since the extensions are unique, the compositions of these  maps in reverse order are the identity on $\deab$ and on $\deabp$. Hence, the convolutor space $C\bigl(\Sab\bigr)$ consists of the same elements of $\dSab$ as  $\deab$ and    $C\bigl(\sab\bigr)$  coincides algebraically with  $\deabp$. We now show  that the extension procedure  used in~\cite{S2012-I,S2020} makes clear the relation between the topologies of these spaces.
 \begin{theorem}
 \label{T4.1}
 The  space  $C\bigl(\sab\bigr)$ of  convolutors for $\sab$ is canonically isomorphic to  $\varprojlim_{A\to 0}\bigl(\mathcal E_{a,A}^{(b)}\bigr)'$.
\end{theorem}
\begin{proof} If $u$ belongs to $C\bigl(\sab\bigr)$, then it can be extended  to a continuous functional on  $\eabp$ in the following way.  Let
$f_0$  be a function in $\sab$ such that $\int
f_0(\xi)d\xi=1$  and let $h\in \eabp$. We set
\begin{equation}
\langle\tilde  u,h\rangle\coloneq\int \langle u,h(\cdot) f_0(\xi-\cdot)\rangle
d\xi\equiv\int(u*h_\xi)(\xi) d\xi,\quad \text{where $h_\xi(x)\coloneq
h(\xi-x)f_0(x)$}.
 \label{4.1}
\end{equation}
 The integrand in~\eqref{4.1} is  well-defined and continuous in $\xi$ because translations act continuously on
$\sab$ and  $h$ is a pointwise multiplier of this space by Theorem~2 in~\cite{S2019-2}. We examine the behaviour of $(u*h_\xi)(\xi)$ at infinity.
The norm $\|f_0\|_{A_0,B_0}$ is finite for any  $A_0,B_0>0$ and there is an $A$ such that $\|h\|^A_B<\infty$ for any $B>0$. Using~\eqref{2.4} and \eqref{2.12}, we obtain
\begin{multline}
|\partial^\beta_x h_\xi(x)|\le \sum_{\gamma\le\beta}\binom{\beta}\gamma
|\partial^\gamma h(\xi-x)\partial^{\beta-\gamma} f_0(x)|\\\le
\|h\|^A_B\,\|f_0\|_{A_0,B_0}\sum_{\gamma\le\beta}\binom{\beta}\gamma
B^{|\gamma|}B_0^{|\beta-\gamma|}b_{|\gamma|}
b_{|\beta-\gamma|}\frac{w_a(|\xi-x|/A)}{w_a(|x|/A_0)}
\\\le K\|h\|^A_B\,\|f_0\|_{A_0,B_0}
(B+B_0)^{|\beta|}b_{|\beta|}\frac {w_a(2|\xi|/A)w_a(2|x|/A)}{w_a(|x|/A_0)}.
 \label{4.2}
\end{multline}
 Choosing
 $A_0\le A/(2H)$ and using~\eqref{2.13}, the $x$-dependent factor in~\eqref{4.2} is estimated as follows
\begin{equation}
\frac{w_a(2|x|/A)}{w_a(|x|/A_0)}\le \frac{w_a(|x|/HA_0)}{w_a(|x|/A_0)}\le \frac{K}{w_a(|x|/HA_0)}.
 \notag
\end{equation}
Consequently, $h_\xi\in S^{b,B+B_0}_{a,HA_0}$ and
\begin{equation}
\|h_\xi\|_{HA_0,B+B_0}\le K\|h\|^A_B\,\|f_0\|_{A_0,B_0} w_a(2|\xi|/A).
 \label{4.3}
\end{equation}
Let $U=\{f\in \sab\colon\|f\|_{A/2H,\,1}<1\}$. Since the map $f\to u\ast f$  is continuous on $\sab$, there is a neighborhood $U_1=\{f\in \sab\colon\|f\|_{A_1,B_1}\le\varepsilon\}$ such that $u*f\in
U$ for all $f\in U_1$. Choose $A_0$, $B_0$ and $B$  such that $HA_0\le A_1$ and $B+B_0\le B_1$. Then $\|f\|_{A_1,B_1}\le \|f\|_{HA_0,B+B_0}$ and~\eqref{4.3} implies that
\begin{equation}
\|u\ast h_\xi\|_{A/2H,\,1}\le \varepsilon^{-1}K\|h\|^A_B\,\|f_0\|_{A_0,B_0} w_a(2|\xi|/A).
 \notag
\end{equation}
Using~\eqref{2.13} again,   we obtain
 \begin{multline}
|(u\ast h_\xi)(\xi)|\le \varepsilon^{-1}K\|h\|^A_B\,\|f_0\|_{A_0,B_0}\frac{w_a(2|\xi|/A)}{w_a(2H|\xi|/A)}\\
\le  \varepsilon^{-1}K^2\|h\|^A_B\,\|f_0\|_{A_0,B_0}\frac{1}{w_a(2|\xi|/A)}.
 \label{4.4}
\end{multline}
Therefore, the integral in~\eqref{4.1} is absolutely convergent and  defines a linear functional $\tilde u$ on $\eabp$.
Since the right-hand side of~\eqref{4.4} contains $\|h\|_B^A$, this functional is continuous on every space $\mathcal E^{(b)}_{a,A}$, $A>0$,
and therefore on $\eabp$.

 Now we show that $\tilde u|_{\sab}=u$.
If $h\in \sab$, then for any positive
 $A_0$,   $B_0$, $A_1$, and $B_1$, we have the inequality
\begin{equation}
|\partial_x^\beta(h(x)f_0(\xi-x))|\le
\|h\|_{A_1,B_1}\|f_0\|_{A_0,B_0}(B_1+B_0)^{|\beta|}b_{|\beta|}\frac{1}{w_a(|x|/A_1)w_a(|\xi-x|/A_0)}.
 \label{4.5}
\end{equation}
  If $HA_1\le A_0$, then  $1/w_a(|x|/A_1)\le K/w_a^2(|x|/A_0)$ and~\eqref{4.5} coupled with~\eqref{2.12} gives
 \begin{equation}
\|(h(\cdot)f_0(\xi-\cdot)\|_{A_0,B_1+B_0}\le K
\|h\|_{A_1,B_1}\|f_0\|_{A_0,B_0}\frac{1}{w_a(|\xi|/2A_0)}.
 \notag
\end{equation}
 We see that in this case, the integral in~\eqref{4.1} remains  absolutely convergent when $u$ is replaced with  any functional in  $\dsab$ and hence the sequence of Riemann sums
    \begin{equation}
s_n(x)=\sum_{\alpha\in\oZ^d, |\alpha|\le n^2}h(x) f_0(\alpha /n-x)/n^d
 \label{4.6}
\end{equation}
  for $\int_{\oR^d} h(x)f_0(\xi-x)d\xi$ is weakly Cauchy in $\sab(\oR^d)$. Since
 $\sab$ is a Montel space, it is   weakly sequentially complete and  weak sequential convergence implies convergence in  its topology.
So, the sequence $s_n(x)$ converges in
   $\sab$. Its limit cannot be anything other than  $h(x)$, because the topology of this space is finer than the topology of simple convergence. Hence,  $\langle \tilde u,h\rangle=\lim_{n\to \infty}\langle u,s_n\rangle=\langle u, h\rangle$ if $h\in\sab$.
Similar arguments show that  $\sab$ is dense in $\eabp$.
Namely, if  $h\in\mathcal E^{(b)}_{a,A}$, then for any positive $A_0$, $B_0$, and $B$ we have
\begin{equation}
|\partial_x^\beta(h(x)f_0(\xi-x))|\le
\|h\|^A_B\|f_0\|_{A_0,B_0}(B+B_0)^{|\beta|}b_{|\beta|}\frac{w_a(|x|/A)}{w_a(|\xi-x|/A_0)},
\label{4.7}
\end{equation}
where $w_a(|x|/A)\le K w_a(H|x|/A)/w_a(|x|/A)$. Choosing $A_0\le A$ and using~\eqref{2.12}, we obtain
 \begin{equation}
\|(h(\cdot)f_0(\xi-\cdot)\|^{A/H}_{B+B_0}\le K
\|h\|^A_B\|f_0\|_{A_0,B_0}\frac{1}{w_a(|\xi|/2A)}.
 \label{4.7*}
\end{equation}
  Hence in this case,  the integral in~\eqref{4.1} is absolutely convergent for any $u$ in the dual of the Montel space
$\mathcal E^{(b)}_{a,A/H-}$ and the sequence~\eqref{4.6}  converges to $h$ in this space and, a fortiori, in $\eabp$.
This proves Proposition~\ref{P3.1}, or more precisely, its part concerning the inclusion  $\sab\hookrightarrow \eabp$ and shows that the continuous extension of $u$ to $\eabp$ is unique.
  In combination with Proposition~\ref{P3.2}, this also completes the proof of the
  algebraic isomorphism between   $C\bigl(\sab\bigr)$ and $\deabp$.

  It remains to show that the map $u\to\tilde u$ from $C\bigl(\sab\bigl)$ to $\varprojlim_{A\to 0}\bigl(\mathcal E_{a,A}^{(b)}\bigr)'$ is continuous, since the continuity of its inverse was proved in the proof of  Proposition~\ref{P3.3}. Let  $\mathscr B_A$ be the set of all bounded subsets of   $\mathcal E^{(b)}_{a,A}$ and $\mathscr B=\bigcup_{A>0}\mathscr B_A$. The polars of sets in $\mathscr B$ form a $0$-neighborhood base for $\varprojlim_{A\to 0}\bigl(\mathcal E_{a,A}^{(b)}\bigr)'$. Therefore it suffices to show that for any given  $G\in\mathscr B_A$, there are a bounded set   $Q$ and a $0$-neighborhood  $V$ in $\sab$ such that $u\ast Q\subset V$ implies  $\sup_{h\in G}|\langle\tilde u,h\rangle|\le 1$. It follows from~\eqref{4.3} that the set $Q=\{w_a(2|\xi|/A)^{-1} h_\xi\colon h\in G, \xi\in\oR^d\}$ is bounded in  $\sab(\oR^d)$. Let
  $V=\{f\in \sab\colon \|f\|_{A/2H,1}< \varepsilon\}$ and let  $\varepsilon^{-1}=K\int_{\oR^d} w_a(2|\xi|/A)^{-1}d\xi$.  If $u\ast Q\subset V$, then
   \begin{equation}
|(u\ast h_\xi)(\xi)|\le \varepsilon\frac{w_a(2|\xi|/A)}{w_a(2H|\xi|/A)}\le  \frac{\varepsilon K}{w_a(2|\xi|/A)}.
 \notag
\end{equation}
Hence, $|\langle\tilde u,h\rangle|\le\int_{\oR^d}|(u\ast h_\xi)(\xi)|d\xi\le 1$,
 which completes the proof.
 \end{proof}

 We turn to the case of spaces $\Sab$. To prove a theorem similar to Theorem~\ref{T4.1}, we need the following simple lemma.

 \begin{lemma} \label{L4.2}
  The weight function $w_a(t)$ satisfies, for any $n\in\oN$, the inequality
\begin{equation}
t^n w_a(t)\le C_n\, w_a(Ht),\quad t\ge 0.
 \label{4.8}
\end{equation}
\end{lemma}
\begin{proof} It follows from the definition~\eqref{2.10}  and from~\eqref{2.3} that
\begin{equation}
t^nw_a(t)=\sup_{k\in\oZ_+}\frac{t^{k+n}}{a_k}\le Ka_n\sup_{k\in\oZ_+}\frac{(Ht)^{k+n}}{a_{k+n}}\le Ka_n\, w_a(Ht).
 \notag
\end{equation}
 Hence,  \eqref{4.8} is satisfied with  $C_n=Ka_n$.
\end{proof}
\begin{theorem}\label{T4.3}
   The  space $C\bigl(\Sab\bigr)$ of convolutors for $\Sab$ is canonically isomorphic to  $\varprojlim_{B\to \infty}\bigl(\mathcal E^{b,B}_{(a)}\bigr)'$.
\end{theorem}
 \begin{proof}
 Every  $u\in C\bigl(\Sab\bigr)$ can be  extended to  a functional  $\tilde u\in
\deab$ in the same manner as given by~\eqref{4.1}, using
$f_0\in \Sab$ with $\int f_0(\xi)d\xi=1$. In this case there are  $A_0,B_0>0$ such that $\|f_0\|_{A_0,B_0}<\infty$ and there is a $B>0$ such that
$\|h\|^A_B<\infty$ for every $A>0$. If  $A\ge 2HA_0$, then the  inequality~\eqref{4.3} holds. The unit ball of $S^{b,B+B_0}_{a,HA_0}$ is bounded in
 $\Sab$ and its image under the continuous map  $f\to u\ast f$ is also bounded in $\Sab$.  Since the inductive limit $\Sab$ is regular, this image is contained and bounded in some
 $S^{b,B_1}_{a,A_1}$, where $A_1$ and  $B_1$  are independent of $A$. Hence there exists a constant  $C>0$ such that
\begin{equation}
\|u\ast h_\xi\|_{A_1,B_1}\le C\|h\|^A_B\,
 w_a(2|\xi|/A).
 \label{4.9}
\end{equation}
In particular, for $A\ge 2HA_1$, we have
 \begin{equation}
|(u\ast h_\xi)(\xi)|\le C\|h\|^A_B\,
\frac{w_a(2|\xi|/A)}{w_a(|\xi|/A_1)}\le  CK\|h\|^A_B\,\frac{1}{w_a(|\xi|/HA_1)}.
 \label{4.10}
\end{equation}
Therefore, the integral in~\eqref{4.1} is absolutely convergent and the functional
 $\tilde u$ is well defined on $\eab$.
Since the right-hand side of~\eqref{4.10} contains $\|h\|_B^A$, this functional is continuous. 

 We now show that
$\tilde u$ coincides with $u$ on $\Sab$. If $h\in S^{b,B_1}_{a,A_1}$, then the inequality~\eqref{4.5} holds and, choosing  $A'\ge \max(HA_1,A_0)$ and using~\eqref{2.12} and~\eqref{2.13}, we obtain
\begin{multline}
\|h(\cdot)f_0(\xi-\cdot)\|_{A',B_1+B_0}\le
\|h\|_{A_1,B_1}\|f_0\|_{A_0,B_0}\sup_x\frac{w_a(|x|/A')}{w_a(|x|/A_1)w_a(|\xi-x|/A_0)}\\
\le K\|h\|_{A_1,B_1}\|f_0\|_{A_0,B_0}\frac{1}{w_a(|x|/A')w_a(|\xi-x|/A_0)}
\le K\|h\|_{A_1,B_1}\|f_0\|_{A_0,B_0} \frac{1}{w_a(|\xi|/2A')}.
 \notag
\end{multline}
Hence, in this case, the integral in~\eqref{4.1}  is absolutely convergent for any  $u\in \dSab$ and the sequence~\eqref{4.6} of Riemann sums converges to $h$ in $\Sab$ by the same argument as for $\sab$ because $\Sab$ is also a Montel space.
In a similar manner we find that  $\Sab$ is dense in  $\eab$. Indeed, if $h\in \mathcal E^{b,B}_{(a)}$, then~\eqref{4.7} holds for any $A>0$ and~\eqref{4.7*} holds for $A\ge A_0$.
Therefore, in this case,  the integral in~\eqref{4.1} is absolutely convergent for any
   $u$ in the dual of the Montel space  $\mathcal E^{b,(B+B_0)+}_{(a)}$, and the sequence~\eqref{4.6} converges to
 $h$ in this space and, a fortiori, in  $\eab$.

 It remains to show that the map $u\to\tilde u$ from $C\bigl(\Sab\bigr)$ to $\varprojlim_{B\to \infty}\bigl(\mathcal E^{b,B}_{(a)}\bigr)'$ is continuous, i.e., that for any $B>0$ and every bounded set $G$ in $\mathcal  E^{b,B}_{(a)}$,  there are a bonded set $Q$  and a $0$-neighborhood  $V$ in $\Sab$ such that  $u\ast Q\subset V$ implies  $\sup_{h\in G}|\langle\tilde u,h\rangle|\le 1$.  If  $A\ge 2HA_0$,  then it follows from~\eqref{4.3} that
 \begin{equation}
\sup_{h\in G}\|h_\xi\|_{HA_0,B+B_0}\le K \|f_0\|_{A_0,B_0}\sup_{h\in G}\|h\|^A_B\, w_a(2|\xi|/A).
 \label{4.11}
\end{equation}
Let  $\mathrm w(\xi)$ be defined by
 \begin{equation}
\mathrm w(\xi)=\inf_{A\ge 2HA_0}\sup_{h\in G}\|h\|^A_B\, w_a(2|\xi|/A).
 \notag
\end{equation}
This function is measurable, locally integrable and bounded from below by a positive constant. It follows from~\eqref{4.11} that the set of functions $h_\xi(x)/\mathrm w(\xi)$, where $h$ runs over $G$ and $\xi$ runs over $\oR^d$,  is bounded in $\Sba(\oR^d)$. We take this set as $Q$ and note that
\begin{equation}
\left\{f\in \Sab(\oR^d)\colon \sup_{x\in \oR^d}\left|f(x)\mathrm w(x)(1+|x|^{d+1})\right|<\epsilon\right\} \qquad (\epsilon>0)
 \label{4.12}
\end{equation}
is a $0$-neighborhood in  $\Sab(\oR^d)$ because by Lemma~\ref{L4.2} we have
 \begin{multline}
\sup_x\left|f(x)\mathrm w(x)(1+|x|^{d+1})\right|\le\sup_{h\in G}\|h\|^A_B\sup_x\left|f(x)(1+|x|^{d+1})w_a(2|x|/A)\right|\le \\ \le \sup_{h\in G}\|h\|^A_B\sup_x\left|f(x)(1+C_{d+1}A^{d+1})w_a(2H|x|/A)\right|\le C'_{d,A}\sup_{h\in G}\|h\|^A_B\,\|f\|_{A/2H,B_1}
 \notag
\end{multline}
 for any $A\ge2HA_0$ and  $B_1>0$; hence the norm $\|f\|_{\mathrm w}=\sup_x|f(x)\mathrm w(x)(1+|x|^{d+1})|$ is weaker than the norm of any space $S_{a,A_1}^{b,,B_1}$, $A_1, B_1>0$. Taking for
 $V$ the neighborhood~\eqref{4.12} with $\epsilon^{-1}=\int_{\oR^d} (1+|\xi|^{d+1})^{-1}d\xi$, we conclude that for each $h\in G$, the inclusion $u\ast Q\subset
 V$ implies
\begin{equation}
|\langle\tilde u,h\rangle|\le \int_{\oR^d}|(u\ast h_\xi)(\xi)|d\xi=\int_{\oR^d}\left|\left(u\ast \frac{h_\xi}{\mathrm w(\xi)}\right)\!(\xi)\,\mathrm w(\xi) \right| d\xi<\epsilon  \int_{\oR^d} \frac{d\xi}{1+|\xi|^{d+1}}= 1.
 \notag
\end{equation}

The proof is somewhat simpler in the particular case $a_n\equiv 1$. Returning to the second line of~\eqref{4.2}, we see that  if $A\ge2A_0$  and $|\xi|\le A/2$, then  $w_1(|\xi-x|/A)=1$  on the support of $1/w_1(|x|/A_0)$. Hence, the inequality~\eqref{4.3} is replaced by $\|h_\xi\|_{A_0, B+B_0}\le \|h\|^A_B \|f_0\|_{A_0,B_0}$, where  $|\xi|\le A/2$. Instead of~\eqref{4.9}, we obtain   $\|u\ast h_\xi\|_{A_1,B_1}\le C\|h\|^A_B$ for all $\xi$ in this region. The function $u\ast h_\xi$ has  compact support in this case and~\eqref{4.10}  holds with $K=1$, $H=1$, the integration in~\eqref{4.1} is over a bounded domain and  $\tilde u$ is  well defined and continuous on $\breve{\mathcal E}^{\{b\}}=\breve{\mathcal E}^{\{b\}}_{(1)}$.  An easy adaptation of the above arguments shows that  $\tilde u=u$ on $S^{\{b\}}_{\{1\}}=\mathcal D^{\{b\}}$ and $\mathcal D^{\{b\}}$ is dense in  $\breve{\mathcal E}^{\{b\}}$. Finally, let $G$ be a bounded set in $\breve{\mathcal E}^{b,B}$. Since $\|h\|^A_B$ increases monotonically with $A$, we have for any $\xi$
 \begin{equation}
\sup_{h\in G}\|h_\xi\|_{A_0,B+B_0}\le  \|f_0\|_{A_0,B_0}\sup_{h\in G}\|h\|^{2A_0+ 2|\xi|}_B.
 \notag
\end{equation}
 Setting $\mathrm  w(\xi)=\sup_{h\in G}\|h\|^{2A_0+ 2|\xi|}_B$, we see that $Q=\{h_\xi(x)/\mathrm w(\xi): h\in G, \xi\in \oR^d\}$ is bounded in $\mathcal D^{\{b\}}(\oR^d)$. The set defined by~\eqref{4.12} is a $0$-neighborhood in this space because $S^{b, B_1}_{\{1\}, A_1}$ consists of functions supported in $\{x\in\oR^d
 \colon |x|\le A_1\}$ and  is continuously embedded into the normed space with the norm $\|\cdot\|_{\mathrm w}$. Taking for $V$  this set, we conclude  as before that $u\ast Q\subset V$ implies $\sup_{h\in G}|\langle\tilde u,h\rangle|\le1$. The proof is complete.  \end{proof}

 \begin{corollary} \label{C4.4}
 The spaces  $C\bigl(\Sab\bigr)$ and  $C\bigl(\sab\bigr)$ are complete and semi-reflexive.
\end{corollary}
This directly follows from the heredity properties of projective limits (Sect.~IV.5.8 in~\cite{Sch}). The same conclusion can be made taking into account that  $C\bigl(\Sab\bigr)$ and  $C\bigl(\sab\bigr)$ are closed subspaces of $\mathcal L\bigl(\Sab\bigr)$ and $\mathcal L\bigl(\sab\bigr)$, respectively, and that $\mathcal L(E)$ is complete and nuclear (and therefore semi-reflexive) for any   nuclear (FS) or (DFS)-space $E$ (Sect.~39.6 in~\cite{K1979} and Sect.~IV.9.7 in~\cite{Sch}).

\begin{corollary} \label{C4.5}  The space  $\eab$ is canonically isomorphic to the strong dual of
 $C\bigl(\Sab\bigr)$ and  $\eabp$ is canonically isomorphic to the strong dual of $C\bigl(\sab\bigr)$.
\end{corollary}
\begin{proof}
    The projective limit $\varprojlim_{B\to \infty}\bigl(\mathcal E^{b,B+}_{(a)}\bigr)'$ is reduced because  $\Sab$ is contained and dense in each  of the spaces $\bigl(\mathcal E^{b,B+}_{(a)}\bigr)'$, $B>0$. Indeed, the map $ \mathcal E^{b,B+}_{(a)}\to \dSab \colon h\to \left(f\to\int h(x)f(x)dx\right)$
     is the adjoint of the natural inclusion $\Sab\to\bigl(\mathcal E^{b,B+}_{(a)}\bigr)'$ with respect to the dualities   $\left\langle \Sab,  \dSab\right\rangle$ and $\left\langle\bigl(\mathcal E^{b,B+}_{(a)}\bigl)', \mathcal E^{b,B+}_{(a)}\right\rangle$  and is injective  because $\Sab$ has sufficiently many functions.
   Therefore $\Sab$ is weakly dense in $\bigl(\mathcal E^{b,B+}_{(a)}\bigr)'$
   and is also strongly dense  because $\bigl(\mathcal E^{b,B+}_{(a)}\bigr)'$ is a  (DFS)-space.
      By Theorem~4.4 in Ch.~IV of~\cite{Sch}, the Mackey dual of
       $\varprojlim_{B\to \infty}\bigl(\mathcal E^{b,B+}_{(a)}\bigr)'$
          is identified with $\varinjlim_{B\to\infty}\bigl(\mathcal E^{b,B+}_{(a)}\bigr)^{\prime\prime}$, i.e., with  $\eab$, because the spaces  $\mathcal E^{b,B+}_{(a)}$ are reflexive.
  Since $\varprojlim_{B\to \infty}\bigl(\mathcal E^{b,B+}_{(a)}\bigr)'$ is semi-reflexive, the strong topology on its dual coincides with the Mackey topology (Theorem~5.5 in Ch.~IV of~\cite{Sch}).
  Hence the strong dual of $\varprojlim_{B\to \infty}\bigl(\mathcal E^{b,B+}_{(a)}\bigr)'$ is identified with $\eab$ and the adjoint of the map $u\to\tilde u$   in Theorem~\ref{T4.3} is an algebraic and topological isomorphism of $\eab$ onto $C^{\,\prime}\bigl(\Sab\bigr)$. In the same way we see that $\eabp$ is isomorphic to          the strong dual of  $C\bigl(\sab\bigr)$.
\end{proof}
The statement of Corollary~\ref{C4.5} is an analog of the well-known fact established by Grothendieck (Sect.~4.4 in  Ch.~II of~\cite{Grot1955}) that the strong dual of the space of convolutors of the Schwartz space $S(\oR^d)$ is isomorphic to the space $\mathcal O_C(\oR^d)$ of very slowly increasing smooth functions.

Under an additional condition on the defining sequences  $a$ and $b$, it follows from Corollary~\ref{C4.5} combined with Theorem~1 in~\cite{S2019} that the spaces  $\eab$ and $\eabp$ are invariant under an important class of ultradifferential operators. Following Komatsu~\cite{K1973}, we write $b_n\subset a_n$ if there exist constants $C$ and $L$ such that
\begin{equation}
  b_n\le C L^n a_n\qquad \forall n\in \oZ_+.
 \label{4.13}
 \end{equation}

 \begin{corollary}\label{C4.6} Let $\mathcal Q=(\mathcal Q^{jk})$ be a  $d\times d$ matrix with real entries and let $b_n\subset a_n$. If the space $\Sab(\oR^d)$ is nontrivial, then the operator $e^{i\mathcal  Q^{jk}\partial_j\partial_k}$ is a homeomorphism of  $\eab(\oR^d)$ onto itself, and if $\sab(\oR^d)$ is nontrivial, then $e^{i\mathcal Q^{jk}\partial_j\partial_k}$ is a homeomorphism of $\eabp(\oR^d)$ onto itself.
\end{corollary}
\begin{proof}
 The operator  $e^{i\mathcal Q^{jk}\partial_j\partial_k}$ is defined via the inverse Fourier transform of $e^{-i\mathcal Q^{jk}\zeta_j\zeta_k}$, as in the case with the Schwartz space $S(\oR^d)$ considered in Sect.~7.6 in~\cite{H1}. In view of Corollary~\ref{C4.5}, the statement of Corollary~\ref{C4.6} amounts to saying that the multiplication by $e^{-i\mathcal Q^{jk}\zeta_j\zeta_k}$  is a self-homeomorphism of the multiplier spaces $M\bigl(\Sba\bigr)$ and $M\bigl(\sba\bigr)$. This in turn follows from Theorem~1 in~\cite{S2019} which states that under condition~\eqref{4.13} the function $e^{-i\mathcal Q^{jk}\zeta_j\zeta_k}$ is a pointwise multiplier of $\Sba$ and  $\sba$.
\end{proof}

  For the
  case of spaces $S_\beta^\alpha$, $\beta\le\alpha$, defined by sequences of the form~\eqref{2.6}, a simple proof of the fact that $e^{-i\mathcal Q^{jk}\zeta_j\zeta_k}$ belongs to $M(S_\beta^\alpha)$ was given  in the course of proving Theorem~1 in~\cite{S2007}. Similar invariance properties of the symbol spaces  $\Gamma_s^\infty$ and $\Gamma_{0,s}^\infty$ (which coincide respectively  with $\breve{\mathcal E}^{\{a\}}_{(a)}$ and  $\breve{\mathcal E}^{(a)}_{\{a\}}$ for $a_n=n^{sn}$, as noted in Sect.~\ref{S2}) were also proved by  different methods in Theorem~4.1 of~\cite{CT}, and in  Proposition~4.4 of~\cite{CW}.

\section{The Moyal multiplier algebras for the spaces of type $S$}
\label{S5}
 Theorem~2 in~\cite{S2019} shows that if~\eqref{4.13} is satisfied, then $\Sab(\oR^{2d})$ and $\sab(\oR^{2d})$ are topological algebras under the twisted multiplication~\eqref{1.1}. ({For $S^\beta_\alpha$, $\beta\le\alpha$,
 this has been proved in~\cite{S2007}.) The Fourier transform
 converts the Weyl-Moyal product into the twisted convolution
 \begin{equation}
(\widehat f\ast_{\hbar}
\widehat g)(\zeta)= \int_{\oR^{2d}}\widehat f(\zeta')
\widehat g(\zeta-\zeta')e^{(i\hbar/2)\zeta\cdot J\zeta'}d\zeta'
 \label{5.1}
\end{equation}
(multiplied by $(2\pi)^{-d}$).
 Under the same condition~\eqref{4.13}, the spaces $\Sba=F\bigl[\Sab\bigr]$ and $\sba=F\bigl[\sab\bigr]$ are topological algebras with the twisted convolution as multiplication.  We will consider them in parallel with $\bigl(\Sab,\star_\hbar\bigr)$ and $\bigl(\sab,\star_\hbar\bigr)$.

 The products  $\langle f\star_\hbar u\rangle$  and $\langle u\star_\hbar f\rangle$  of a function  $f\in \Sab$   with an element $u$ of the dual space $\dSab$ are defined by
\begin{equation}
\langle f\star_\hbar u,h\rangle\coloneq \langle u,h\star_\hbar f\rangle ,\qquad  \langle u\star_\hbar f,h\rangle\coloneq \langle u,f\star_\hbar h\rangle\qquad \forall h\in \Sab
\label{5.2}
\end{equation}
and analogously for the dual pair  $\bigl\langle\sab,  \dsab \bigr\rangle$.
Since the right-hand sides in~\eqref{5.2} are linear and continuous in $h$, these products are well defined as elements of $\dSab$.   The formulas~\eqref{5.2}  agree with the definition of the  operation $\star_\hbar$ in $\Sab$ due to the identity
 \begin{equation}
\int(f\star_\hbar g)(x)dx=\int f(x)g(x)dx,
\label{5.3}
\end{equation}
which is equivalent to the obvious identity  $(\widehat f\ast_\hbar\widehat g)(0)=(\widehat f\ast\widehat g)(0)$. Indeed, if $f,g,h\in \Sab$, then using~\eqref{5.3} and the associativity of the Weyl-Moyl product, we obtain
\begin{equation}
\langle f\star_\hbar g, h\rangle\equiv\int(f\star_\hbar g)(x) h(x)dx= \int(h\star_\hbar(f\star_\hbar g))(x) dx =
\int((h\star_\hbar f)\star_\hbar g)(x) dx=\langle g,h\star_\hbar f\rangle.
 \notag
\end{equation}
Under condition~\eqref{4.13}, the twisted convolution product of $g\in \Sba$ and $v\in \dSba$ can also be defined by duality, namely:
\begin{equation}
\langle v\ast_\hbar g,h\rangle\coloneq \langle v,\check{g}\ast_\hbar h\rangle,\quad \langle g\ast_\hbar v,h\rangle\coloneq \langle v,h\ast_\hbar\check{g} \rangle \qquad \forall h\in \Sba,
\notag
\end{equation}
where   $\check{g}(\zeta)=g(-\zeta)$.  Then we clearly have the  relations
\begin{equation}
\widehat{u\star_\hbar f} = (2\pi)^{-d}\widehat{u}\ast_\hbar\widehat{f},\qquad
\widehat{f\star_\hbar u} = (2\pi)^{-d}\widehat{f}\ast_\hbar\widehat{u}.
\notag
\end{equation}
 The spaces of left and right multipliers for the algebra $(\Sab,\star_\hbar)$ are defined as follows:
\begin{gather}
\mM_{\hbar,L}\bigl(\Sab\bigr)\coloneq \left\{u\in \dSab\colon
u\star_\hbar f\in \Sab\quad \forall  f\in \Sab\right\},\label{5.4}
\\  \mM_{\hbar,R}\bigl(\Sab\bigr)\coloneq\left \{u\in
 \dSab\colon  f\star_\hbar u
\in \Sab\quad \forall f\in\Sab\right\},
\label{5.5}
 \end{gather}
 The definitions of $\mM_{\hbar,L}\bigl(\sab\bigr)$ and  $\mM_{\hbar,R}\bigl(\sab\bigr)$ are similar. The mappings $f\to u\star_\hbar f$ and $f\to f\ast_\hbar u$ of $\Sab$ into itself and of  $\sab$ into itself are continuous by the closed graph theorem, and the multiplier spaces are naturally endowed with the respective topologies induced by  $\mathcal L\bigl(\Sab\bigr)$  and $\mathcal L\bigl(\sab\bigr)$. Theorem~3 in~\cite{S2019}  establishes that under condition~\eqref{4.13}, the spaces ${\mM}_{\hbar,L}\bigl(\Sab\bigr)$,  ${\mM}_{\hbar,R}\bigl(\Sab\bigr)$, ${\mM}_{\hbar,R}\bigl(\sab\bigr)$,  and ${\mM}_{\hbar,L}\bigl(\sab\bigr)$ are unital  algebras with separately continuous multiplication $\star_\hbar$. (For the case of
  spaces $S^\beta_\alpha$, $\beta\le\alpha$, this was proved in ~\cite{S2011}.) The Fourier transformation maps ${\mM}_{\hbar,L}\bigl(\Sab\bigr)$ and  ${\mM}_{\hbar,R}\bigl(\Sab\bigr)$
  respectively onto the algebras of left and right twisted convolution  multipliers
 \begin{gather}
\mathcal C_{\hbar,L}\bigl(\Sba\bigr)= \left\{v\in \dSba\colon
v\ast_\hbar g\in \Sba\quad \forall  g\in \Sba\right\}, \label{5.6}
\\
  \mathcal C_{\hbar,R}\bigl(\Sba\bigr)=\left \{v\in
 \dSba\colon  g\ast_\hbar v
\in \Sba\quad \forall g\in \Sba\right\}.
\label{5.7}
 \end{gather}
Analogously,  $\widehat{\mM}_{\hbar,L}\bigl(\sab\bigr)=\mathcal C_{\hbar,L}\bigl(\sba\bigr)$ and $\widehat{\mM}_{\hbar,R}\bigl(\sab\bigr)=\mathcal C_{\hbar,R}\bigl(\sba\bigr)$.

 \begin{lemma} \label{L5.1} Let $b_n\subset a_n$. The algebras $\bigl(\Sab,\star_\hbar\bigr)$, $\bigl(\sab,\star_\hbar\bigr)$, $\bigl(\Sba,\ast_\hbar\bigr)$, and $\bigl(\sba,\ast_\hbar\bigr)$,   have sequential approximate identities.
\end{lemma}
\begin{proof} Since the algebra $\bigl(\Sab,\star_\hbar\bigr)$ is isomorphic, via the Fourier transform, to  the algebra $\bigl(\Sba,\ast_\hbar\bigr)$ and $\bigl(\sab,\star_\hbar\bigr)$  is isomorphic to $\bigl(\sba,\ast_\hbar\bigr)$, it suffices to consider the case of twisted convolution.
  Every nontrivial space $\Sba(\oR^{2d})$ contains a function $e_1(\zeta)$ such that $\int_{\oR^{2d}} e_1(\zeta)d\zeta=1$ and $e_1(\zeta)\ge 0$, because it is an algebra under pointwise multiplication.  We claim that the sequence $e_n(\zeta)= n^{2d} e_1(n\zeta)$, $n\in\oN$, is an approximate identity for $\bigl(\Sba,\ast_\hbar\bigr)$, i.e., that for any $g\in \Sba$, the limit relations  $e_n\ast_\hbar g\to g$ and  $g\ast_\hbar e_n\to g$ hold in the topology of $\Sab$ as $n\to\infty$. Clearly, $(e_n\ast_\hbar g)(\zeta)\to g(\zeta)$ and  $(g\ast_\hbar e_n)(\zeta)\to g(\zeta)$ at every $\zeta$ because $e_n$ is a delta-like sequence. We show that the sequence  $e_n\ast_\hbar g$ is bounded in  $\Sba$. Using~\eqref{2.4} and  the inequality $t^k\le {A'}^k a_k\, w_a(t/A')$, valid for any $A'>0$, and also the inequality $w_a(t)\le Cw_b(Lt)$ following from~\eqref{4.13}, we obtain
\begin{multline}
|\partial^\alpha(e_n\ast_\hbar g)(\zeta)|\le
\int_{\oR^{2d}}e_n(\zeta')\left|\partial^\alpha_\zeta \left(e^{(i\hbar/2)\zeta\cdot J\zeta'} g(\zeta-\zeta')\right)\right|d\zeta'\le \\
\le  \|g\|_{A,B}\sum_{\gamma\le\alpha}{\alpha\choose\gamma} A^{|\alpha-\gamma|}a_{|\alpha-\gamma|}
\int_{\oR^{2d}}\frac{e_n(\zeta')\,(\hbar|\zeta'|/2)^{|\gamma|}}{w_b(|\zeta-\zeta'|/B)}d\zeta' \le \\ \le C \|g\|_{A,B} (A+A'\hbar/2)^{|\alpha|} a_{|\alpha|} \int_{\oR^{2d}}\frac{e_n(\zeta')\,w_b(L|\zeta'|/A')}{w_b(|\zeta-\zeta'|/B)}d\zeta'.
 \label{5.8}
 \end{multline}
  The function $e_1$ belongs to  $S^{a, A_1}_{b,B_1/H}$ with sufficiently large $A_1$ and $B_1$ and   by~\eqref{2.11} we have
  \begin{equation}
\frac{1}{w_b(|\zeta-\zeta'|/B)}\le\frac{w_b(|\zeta'|/(HB_1))}{w_b(|\zeta|/(B+HB_1))}.
\label{5.9}
 \end{equation}
 Let $A'=LHB_1$. Then it follows from~\eqref{5.8}, \eqref{5.9} and~\eqref{2.13} that
  \begin{equation}
|\partial^\alpha(e_n\ast_\hbar g)(\zeta)|\le  CK \|g\|_{A,B} \frac{(A+LHB_1\hbar/2)^{|\alpha|} a_{|\alpha|} }{w_b(|\zeta|/(B+HB_1))}\int_{\oR^{2d}} e_n(\zeta')\,w_b(|\zeta'|/B_1)d\zeta'.
\label{5.10}
 \end{equation}
 Since
  $e_n(\zeta')\le C_n/ w_b(H|\zeta'|/B_1)$, the integral in the right-hand side is finite for any $n$ in view~\eqref{2.13}
   and tends to  $w_b(0)=1$ as $n\to\infty$. Therefore, the sequence  $e_n\ast_\hbar g$ is contained and  bounded in $S^{a, A_2}_{b,B_2}$, where $A_2=A+LHB_1\hbar/2$ and  $B_2=B+HB_1$. This sequence
    has at least one limit point in the Montel space $\Sba$.
 The topology of $\Sba$  is finer than  the topology  of simple convergence, and only $g$ can be the limit point.  Hence $e_n\ast_\hbar g$ converges to  $g$  in $\Sba$, because otherwise it would have a limit point other than $g$, which is impossible. Similarly $g\ast_\hbar e_n\to g$ in $\Sba$. The proof is  simpler in the case where $b_n\equiv 1$. Then $e_1\in\mathcal D^{\{a\}}=S^{\{a\}}_{\{1\}}$ is supported in a compact set $\{x \colon|x|\le B_1\}$ and the integral in the second line of~\eqref{5.8} is clearly less than $(B_1\hbar/2)^{|\gamma|}a_{|\gamma|} \int_{|\zeta-\zeta'|\le B} e_n(\zeta')d\zeta'$,
 which immediately implies the boundedness of the sequence $e_n\ast_\hbar g$ in $\mathcal D^{\{a\}}$.

  In the case of spaces $\sba$, the norm  $\|g\|_{A,B}$ is finite for any $A,B>0$,  and $e_1\in S^{a, A_1}_{b,B_1/H}$ for arbitrarily small $A_1$ and $B_1$. Hence, the same estimate~\eqref{5.10} shows that the sequence $e_n\ast_\hbar g$ is bounded in $\sba$.  The rest of the proof is the same as for $\Sba$, because $\sba$ is also a Montel space. \end{proof}

The above proof generalizes and simplifies a proof given in~\cite{S2011} for $(S^\beta_\alpha,\star_\hbar)$, $\beta\le\alpha$.

 \begin{theorem}\label{T5.2}  The algebra  $\mM_{\hbar,L}\bigl(\Sab\bigr)$  is canonically identified with the closure in $\mathcal L\bigl(\Sab\bigr)$ of the set of all operators of the left $\star_\hbar$-multiplication by elements of $\Sab$. This closure consists of  all $V\in\mathcal L\bigl(\Sab\bigr)$ such that
\begin{equation}
V(f\star_\hbar g)=V(f)\star_\hbar g\qquad \forall f,g\in \Sab.
\label{5.11}
 \end{equation}
   A similar statement holds for $\mM_{\hbar,R}\bigl(\Sab\bigr)$, with the replacement of  the left $\star_\hbar$-multiplication by the right $\star_\hbar$-multiplication and with
   the condition
   $V(f\star_\hbar g)=f\star_\hbar V(g)$ instead of~\eqref{5.11}. Analogous statements are true for $\mM_{\hbar,L}\bigl(\sab\bigr)$ and $\mM_{\hbar,R}\bigl(\sab\bigr)$.
   \end{theorem}

\begin{proof} Let  $u\in \mM_{\hbar,L}\bigl(\Sab\bigr)$ and let  $L_u$ denote the map $f\to u\star f$ from $\Sab$ to itself. (For brevity, we temporarily suppress the index $\hbar$ in the notation of the star product).
  Using an approximate identity $e_n$ for $\bigl(\Sab,\star\bigr)$ and passing to the limit in  $\langle u,e_n\star f\rangle=\langle u\star e_n,f\rangle$, we find that the map $u\to L_u$ of $\mM_{\hbar,L}\bigl(\Sab\bigr)$ to  $\mathcal L\bigl(\Sab\bigr)$ is injective. It follows from the definition~\eqref{5.2} and from the associativity of the $\star$-multiplication in $\Sab$ that \begin{equation}
 u\star(f\star g)=(u\star f)\star g,\qquad (f\star g)\star u=f\star(g\star u).
 \label{5.12}
 \end{equation}
  In terms of  $L_u$, the first of these relations takes the form $L_u(f\star g)=L_u(f)\star g$.
On the other hand, to every  $V\in \mathcal L\bigl(\Sab\bigr)$ there corresponds  a unique
 $v\in \dSab$ such that
\begin{equation}
\langle v,f\rangle=\int V(f)\,dx.
 \notag
 \end{equation}
If $V$ satisfies~\eqref{5.11}, then~\eqref{5.2} and~\eqref{5.3} give
\begin{equation}
\langle v\star f, g\rangle=\langle v,f\star g\rangle=\int V(f\star g)\,dx=\int V(f) g\,dx.
 \notag
 \end{equation}
Hence, $v\star f=V(f)$, $v\in\mM_{\hbar,L}\bigl(\Sab\bigr)$, and $L_v=V$. The relation $V(e_n\star f)= V(e_n)\star f$ implies that the sequence  of  operators of the left
$\star$-multiplication by $V(e_n)\in \Sab$ converges pointwise  to
 $V$. Since   $\Sab$ is barrelled, the convergence is uniform on every precompact subset of $\Sab$ by the Banach-Steinhaus theorem  (Theorem 4.6 in Ch.~III in~\cite{Sch}). Since, further,   $\Sab$ is a Montel space, every its bounded subset is precompact, and we conclude that the sequence in question converges to  $V$ in the topology of $\mathcal L\bigl(\Sab\bigr)$. On the other hand, if
  $V\in \mathcal L\bigl(\Sab\bigr)$ is the limit of a net of operators of the left
 $\star$-multiplication by $h_\nu\in \Sab$, then
\begin{equation}
V(f\star g)=\lim_\nu h_\nu\star(f\star g)=\lim_\nu
(h_\nu\star f)\star g= V(f)\star g
 \notag
 \end{equation}
 and hence $V$ satisfies~\eqref{5.2}. The case of right multipliers is treated similarly,  using the second of relations~\eqref{5.12}. The same arguments apply to  $\mM_{\hbar,L}\bigl(\sab\bigr)$ and $\mM_{\hbar,R}\bigl(\sab\bigr)$, which completes the proof.
\end{proof}

 \begin{remark}\label{R5.3} Similar theorems hold  for  $\mathcal C_{\hbar,L}\bigl(\Sba\bigr)$,  $\mathcal C_{\hbar,L}\bigl(\sba\bigr)$, $\mathcal C_{\hbar,R}\bigl(\Sba\bigr)$, and  $\mathcal C_{\hbar,R}\bigl(\sba\bigr)$. These algebras can also be characterized in another way.
 Proposition~2 in~\cite{S2012-I} shows that  $\mathcal C_{\hbar,L}\bigl(\Sba\bigr)$ and  $\mathcal C_{\hbar,L}\bigl(\sba\bigr)$ can be identified with the respective sets of those operators in $\mathcal L\bigl(\Sba\bigr)$ and in $\mathcal L\bigl(\sba\bigr)$ that commute with the twisted translations $\tau_\xi\colon g(\zeta)\to e^{(i\hbar/2)\xi\cdot J\zeta}g(\zeta-\xi)$, $\xi\in \oR^{2d}$. Analogous statements are valid for $\mathcal C_{\hbar,R}\bigl(\Sba\bigr)$ and  $\mathcal C_{\hbar,R}\bigl(\sba\bigr)$, but  with  $\bar\tau_\xi\colon g(\zeta)\to e^{-(i\hbar/2)\xi\cdot J\zeta}g(\zeta-\xi)$ in place of $\tau_\xi$. Theorem~\ref{T5.2} and the above characterizations  hold  true at  $\hbar=0$, i.e., in the case of pointwise multiplication and ordinary convolution.
\end{remark}

In the sequel, we consider the spaces of two-sided multipliers
\begin{gather}
\hspace{-1.1mm}\mM_\hbar\bigl(\Sab\bigr)=\mM_{\hbar,L}\bigl(\Sab\bigr)\bigcap \mM_{\hbar,R}\bigl(\Sab\bigr),\quad\!\!\!
\mM_\hbar\bigl(\sab\bigr)=\mM_{\hbar, L}\bigl(\sab\bigr)\bigcap\mM_{\hbar, R}\bigl(\sab\bigr)\label{5.13} \\
{\mathcal C}_\hbar\bigl(\Sba\bigr)={\mathcal C}_{\hbar,L}\bigl(\Sba\bigr)\bigcap
{\mathcal C}_{\hbar,R}\bigl(\Sba\bigr),\quad
{\mathcal C}_\hbar\bigl(\sba\bigr)={\mathcal C}_{\hbar, L}\bigl(\sba\bigr)\bigcap{\mathcal C}_{\hbar, R}\bigl(\sba\bigr).
\label{5.14}
 \end{gather}
  The space $\mM_\hbar\bigl(\Sab\bigr)$ is naturally endowed with the initial topology with respect to the inclusion maps $\mM_\hbar\bigl(\Sab\bigr)\to \mM_{\hbar,L}\bigl(\Sab\bigr)$ and $\mM_\hbar\bigl(\Sab\bigr)\to
  \mM_{\hbar,R}\bigl(\Sab\bigr)$. The spaces $\mM_\hbar\bigl(\sab\bigr)$, ${\mathcal C}_\hbar\bigl(\Sba\bigr)$, and ${\mathcal C}_\hbar\bigl(\sba\bigr)$ are topologized in the same manner. All these spaces are unital involutive algebras with separately continuous multiplication (for more detail, see~\cite{S2019} and also~\cite{S2011} for the case of spaces $S_\alpha^\beta$).
  We note that  multiplication in  $\mM_\hbar\bigl(\Sab\bigr)$ and in $\mM_\hbar\bigl(\sab\bigr)$ can be defined by either of the two formulas
\begin{equation}
\langle u\star v,f\rangle\coloneq\langle u, v\star f\rangle,\qquad
\langle u\star v,f\rangle\coloneq\langle v,f\star u\rangle.
  \notag
 \end{equation}
Indeed, replacing
 $f$ by $e_n\star f$ and using~\eqref{5.12} and then~\eqref{5.2}, we can write their right-hand sides  as
 $\int (f\star u)(v\star e_n) dx$. Passing to the limit as
$n\to\infty$ and using  the continuity of the maps $f\to v\star f$
and $f\to  f\star u$, we see that these  definitions are equivalent.

\section{Inclusion relations between the Moyal multiplier algebras  and  spaces of type $\mathscr E$}
\label{S6}
Along with the inclusions~\eqref{3.1},  we  have the  continuous inclusions
\begin{equation}
 \Sba\hookrightarrow \Eba,\qquad \sba \hookrightarrow  \Ebap,
  \label{6.1}
 \end{equation}
 where $\Eba$ and $\Ebap$ are defined by~\eqref{2.15}.
 (Recall that the upper index determines the smoothness of the space elements, and the lower index determines their behavior at infinity.)  Lemmas~3 and 4 in~\cite{S2019-2} show that these inclusions are dense.
  Therefore, $\dEba$ and $\dEbap$ are naturally identified with the respective vector subspaces of  $\dSba$ and $\dsba$. It follows from~\eqref{2.17} that $\dEba\subset\deba$ and $\dEbap\subset\debap$. The noncommutative deformation of convolution  violates the inclusion relations~\eqref{3.5}, but Theorem~4 in~\cite{S2019-2} shows that under condition~\eqref{4.13}, the  inclusions
 \begin{equation}
 \dEba\subset\mathcal C_\hbar\bigl(\Sba\bigr),\qquad \dEbap \subset \mathcal C_\hbar\bigl(\sba\bigr)
  \label{6.2}
 \end{equation}
 are valid. They are the starting point for  deriving other inclusion relationships in this section.

 \begin{theorem}\label{T6.1} The inclusions~\eqref{6.2} are continuous.
\end{theorem}
\begin{proof} As in the case of spaces~\eqref{2.16}, it is useful to represent  $\Eba$ and $\Ebap$ as limits of families of spaces with nice topological properties, namely:
    \begin{equation}
 \Eba=\varprojlim_{B\to\infty}\mathcal E^{\{a\}}_{b,B+},\qquad \Ebap=\varprojlim_{A\to0}  \mathcal E^{a,A-}_{\{b\}},
  \label{6.3}
 \end{equation}
 where
\begin{equation}
\mathcal E^{\{a\}}_{b,B+}\coloneq\varinjlim\limits_{A\to\infty,\epsilon\to0}\mathcal E_{b,B+\epsilon}^{a,A},
\qquad
\mathcal E^{a,A-}_{\{b\}}\coloneq\varinjlim\limits_{B\to0,\epsilon\to0}\mathcal E^{a,A-\epsilon}_{b,B}  .
\notag
 \end{equation}
 By Lemma~2 in~\cite{S2019-2}, $\mathcal E^{\{a\}}_{b,B+}$ and $\mathcal E^{a,A-}_{\{b\}}$ are (DFS)-spaces.
    It follows that the projective limits~\eqref{6.3} are semi-reflexive.
  Furthermore, Lemmas~3 and 4 in~\cite{S2019-2} show that they are reduced.
   The duals $\bigl(\mathcal E^{\{a\}}_{b,B+}\bigr)'$ and $\bigl(\mathcal E^{a,A-}_{\{b\}}\bigr)'$ are (FS)-spaces and therefore Mackey spaces.
    Using, as in the proof of Corollary~\ref{C4.5}, the duality between projective and inductive limits,     we conclude that
   \begin{equation}
\dEba= \varinjlim\limits_{B\to\infty}\bigl(\mathcal E^{\{a\}}_{b,B+}\bigr)',\qquad
\dEbap= \varinjlim\limits_{A\to0} \bigl(\mathcal E^{a,A-}_{\{b\}}\bigr)' ,
\notag
 \end{equation}
because the semi-reflexivity of  $\dEba$ and  $\dEbap$ implies that the strong topology on $\dEba$ and  $\dEbap$ coincides with the Mackey topology.
   It suffices to show now that the maps
   $\bigl(\mathcal E^{\{a\}}_{b,B+}\bigr)'\to {\mathcal C}_\hbar\bigl(\Sba\bigr)$  and  $\bigl(\mathcal E^{a,A-}_{\{b\}}\bigr)'\to {\mathcal C}_\hbar\bigl(\sba\bigr)$ are continuous for every $B>0$ and every $A>0$. We note that for any fixed $g\in \Sba$, the graphs of the maps
   \begin{equation}
 \bigl(\mathcal E^{\{a\}}_{b,B+}\bigr)'\to \Sba\colon v\to v\ast_\hbar g ,\qquad  \bigl(\mathcal E^{\{a\}}_{b,B+}\bigr)'\to  \Sba\colon v\to g\ast_\hbar v
\label{6.4}
 \end{equation}
  are closed. Indeed, if $v_\nu$ is a net in $\bigl(\mathcal E^{\{a\}}_{b,B+}\bigr)'$ such that $v_\nu\to v\in\bigl(\mathcal E^{\{a\}}_{b,B+}\bigr)'$ and $v_\nu\ast_\hbar g\to f\in \Sba$, then for any $h\in \Sba$, we have
   \begin{equation}
 \int f(\zeta) h(\zeta)d\zeta= \lim_\nu \langle v_\nu\ast_\hbar g, h\rangle= \lim_\nu\langle v_\nu,\check{g}\ast_\hbar
h\rangle=\langle v, \check{g}\ast_\hbar
h\rangle=\langle v\ast_\hbar g, h\rangle,
 \notag
 \end{equation}
hence   $f=v\ast_\hbar g$. Consequently, the maps~\eqref{6.4} are continuous by the closed graph theorem. The rest of the proof is similar to that of Theorem~\ref{T5.2}.
Let $v_n$ be a sequence in $\bigl(\mathcal E^{\{a\}}_{b,B+}\bigr)'$ converging to zero. Then $v_n\ast_\hbar g\to0$ and  $g\ast_\hbar v_n\to0$ for any $g\in  \Sba$ and the convergence is  uniform on every precompact subset of $\Sba$ by the Banach-Steinhaus theorem.  Since $\Sba$ is a Montel space, every its bounded subset is precompact. Hence, $v_n\to 0$ in the topology of ${\mathcal C}_\hbar\bigl(\Sba\bigr)$. Analogous arguments show that the second of inclusions~\eqref{6.2} is also continuous, which completes the proof. \end{proof}

It is clear from the definitions~\eqref{2.14}-\eqref{2.16} that the Palamodov spaces $\Eab$ and $\eab$  are  naturally embedded in $\dSab$, whereas  $\mathcal E^{(b)}_{\{a\}}$ and $\breve{\mathcal E}^{(b)}_{\{a\}}$ are naturally embedded in $\dsab$.

 \begin{theorem} \label{T6.2} If $b_n\subset a_n$, then the following continuous embeddings hold:
 \begin{equation}
 \eab\hookrightarrow {\mathscr M}_\hbar\bigl(\Sab\bigr),\qquad \eabp\hookrightarrow {\mathscr M}_\hbar\bigl(\sab\bigr).
\label{6.5}
 \end{equation}
\end{theorem}
 \begin{proof}  Theorem~2 in~\cite{S2019-2} shows that the functions of  $\Eba$ are pointwise multipliers for $\Sba$, the functions of $\Ebap$ are pointwise multipliers for  $\sba$, and the inclusion maps
   \begin{equation}
  \Eba\to M\bigl(\Sba\bigr),\qquad  \Ebap\to M\bigl(\sba\bigr),
\label{6.6}
\end{equation}
 are continuous. These maps clearly have dense ranges (and are even surjective, but we will not prove this here), and their adjoints
 \begin{equation}
M'\bigl(\Sba\bigr) \to  \dEba ,\qquad M'\bigl(\sba\bigr) \to  \dEbap.
\label{6.7}
\end{equation}
are  continuous and injective.
The compositions of the canonical inclusions $\dEba\hookrightarrow\dSba$ and  $\dEbap\hookrightarrow\dsba$ with their respective maps in~\eqref{6.7} are the adjoints of the canonical inclusions  $\Sba\hookrightarrow M\bigl(\Sba\bigr)$ and
$\sba\hookrightarrow M\bigl(\sba\bigr)$. If $b_n\subset a_n$, then~\eqref{6.7} in combination with Theorem~\ref{T6.1} gives
 \begin{equation}
M'\bigl(\Sba\bigr) \hookrightarrow  \mathcal C_\hbar\bigl(\Sba\bigr) ,\qquad M'\bigl(\sba\bigr) \hookrightarrow   \mathcal C_\hbar\bigl(\sba\bigr).
\label{6.8}
\end{equation}
After Fourier transforming we obtain
 \begin{equation}
C'\bigl(\Sab\bigr) \hookrightarrow  \mathscr M_\hbar\bigl(\Sab\bigr) ,\qquad C'\bigl(\sab\bigr) \hookrightarrow   \mathscr M_\hbar\bigl(\sab\bigr).
\label{6.9}
\end{equation}
By Corollary~\ref{C4.5},  $\eab$ is isomorphic to  $C'\bigl(\Sab\bigr)$ and this isomorphism is implemented by the adjoint of the map
 $\mathbf i\colon u\to \tilde u$ from Theorem~\ref{T4.3}. Let  $\mathbf j$  denote the natural embedding of $\Sab$ into $C\bigl(\Sab\bigr)$.
    The composition $\mathbf j'\circ \mathbf i'$ of the adjoint maps is
   precisely
   the natural embedding of $\eab$ into $\dSab$ because
$\langle (\mathbf j'\circ\mathbf i') (h), f\rangle=\langle h, \mathbf i(\mathbf j(f))\rangle$ for any $h\in \eab$ and  $f\in \Sab$ and by the definition~\eqref{4.1} we have for $u= \mathbf j (f)$
   \begin{equation}
   \langle \mathbf i(\mathbf j (f)), h\rangle=\int\!\int f(x)h(x)f_0(\xi-x)dxd\xi =\int f(x)h(x)dx.
   \notag
   \end{equation}
 Similarly, the second of embeddings~\eqref{6.9}, together with Corollary~\ref{C4.5}, implies the second of  embeddings~\eqref{6.5}. \end{proof}

\begin{remark}\label{R6.3}
  Unlike $\eab$, the space $\Eab$ is not contained in $\mathscr M_\hbar\bigl(\Sab\bigr)$ and $\Eabp$ is not contained in $\mathscr M_\hbar\bigl(\sab\bigr)$. In particular, the function
 $e^{(2i/\hbar)p\cdot q}$, where $(p,q)$ are symplectic coordinates on $\oR^{2d}$, belongs to $\mathcal E^{\{b\}}_{(a)}(\oR^{2d})$, but does not belong to $\mathscr M_\hbar\bigl(S_{\{a\}}^{\{b\}}(\oR^{2d})\bigr)$, see Proposition~6 in~\cite{S2012-II}.
\end{remark}
In the important case of the Fourier-invariant spaces of type $S$, we obtain  the following additional result.

 \begin{corollary}\label{C6.4} If $\hbar\ne0$,  then $\breve{\mathcal E}^{\{a\}}_{(a)}$ is contained in ${\mathcal C}_\hbar\bigl(S_{\{a\}}^{\{a\}}\bigr)$ and $\breve{\mathcal E}^{(a)}_{\{a\}}$ is contained in ${\mathcal C}_\hbar\bigl(S_{(a)}^{(a)}\bigr)$ with continuous inclusions.
\end{corollary}
\begin{proof}  Here we use the symplectic Fourier transforms defined by
\begin{equation}
(F_J f)(y)\coloneq
(\pi\hbar)^{-d}\int\limits_{\oR^{2d}}  f(x)\,e^{-(2i/\hbar)x\cdot Jy} dx,\quad
(\bar F_J f)(y)\coloneq
(\pi\hbar)^{-d}\int\limits_{\oR^{2d}}  f(x)\,e^{(2i/\hbar)x\cdot Jy} dx.
\notag
\end{equation}
Performing the integration  over one of the variables $x'$ or $x''$ in~\eqref{1.1}, we obtain \begin{multline}
(f\star_\hbar g)(x)=
(\pi\hbar)^{-d}\int\limits_{\oR^{2d}}\! (F_Jf)(x'') g(x- x'')\,e^{(2i/\hbar)x\cdot Jx''} dx''=
 (\pi\hbar)^{-d}\big((F_J f)\ast_{4/\hbar} g\big)(x) \\
=(\pi\hbar)^{-d}\int\limits_{\oR^{2d}}\!  f (x- x') (\bar F_J g)(x') \,e^{(2i/\hbar)x'\cdot Jx} dx'= (\pi\hbar)^{-d}
\big( f\ast_{4/\hbar}(\bar F_J g)\big)(x).
\label{6.10}
\end{multline}
 The maps $v\to f\star v$ and  $v\to v\star g$ of  $S_{\{a\}}^{\{a\}\prime}$ into itself are the adjoints of the respective continuous maps
$h\to h\star f$ and $h\to g\star h$ of $S_{\{a\}}^{\{a\}}$ into itself and are therefore continuous. The twisted convolution products in~\eqref{6.10} have similar continuity properties. Furthermore, $S_{\{a\}}^{\{a\}}$ is dense in $S_{\{a\}}^{\{a\}\prime}$. Hence, it follows from~\eqref{6.10} that
\begin{equation}
f\star_\hbar v= (\pi\hbar)^{-d}(F_J f)\ast_{4/\hbar} v,\quad
v\star_\hbar g=(\pi\hbar)^{-d} v\ast_{4/\hbar}
(\bar F_J g)\quad\forall f,g\in S_{\{a\}}^{\{a\}},\, v\in S_{\{a\}}^{\{a\}\prime}.
 \notag
\end{equation}
  If $v\in \breve{\mathcal E}^{\{a\}}_{(a)}$, then $f\star_\hbar v\in S_{\{a\}}^{\{a\}}$ and $v\star_\hbar g\in S_{\{a\}}^{\{a\}}$ for any $\hbar$ by Theorem~\ref{T6.2}. Because $F_J$ and $\bar F_J$  are automorphisms of $S_{\{a\}}^{\{a\}}$, we conclude that  $\breve{\mathcal E}^{\{a\}}_{(a)}\subset  {\mathcal C}_\hbar\bigl(S_{\{a\}}^{\{a\}}\bigr)$ for $\hbar\ne0$. The reasoning used in proving Theorem~\ref{T6.1} shows that this inclusion is continuous. The proof for the case of  $S_{(a)}^{(a)}$  is entirely analogous.
\end{proof}
Theorem~\ref{T6.2} can be extended to the multiplier algebras associated with the quantization map
 \begin{equation}
f \longmapsto \Op_{S}(f)=(2\pi)^{-d}\int_{\oR^{2d}}\! \widehat f(\zeta) e^{(i\hbar/4)\zeta\cdot S\zeta} T_\zeta^\hbar  d\zeta,
 \label{6.11}
\end{equation}
 where $S$ is a real symmetric matrix and $T_\zeta^\hbar$ is a Weyl system of unitary operators satisfying
\begin{equation}
T^\hbar_\zeta T^\hbar_{\zeta'}= e^{-(i\hbar/2)\zeta\cdot
J\zeta'}T^\hbar_{\zeta+\zeta'}.
 \label{6.12}
\end{equation}
 This generalization covers, in particular, various operator orderings that differ from the  Weyl (totally symmetric) ordering corresponding to  $S=0$ when we use the notation $\Op(f)$.
 The quantization map~\eqref{6.11} implies the following composition law for symbols
 \begin{equation}
(f\star_{\hbar, S} g)(x)=(\pi\hbar)^{-2d}\int_{\oR^{4d}}
f\left(x-(J+ S)x'\right)g(x-x'')e^{(2i/\hbar)x'\cdot x''}dx'dx''.
\label{6.13}
\end{equation}
The Fourier transform
converts~\eqref{6.13} into the deformed convolution
\begin{equation}
(\widehat f\ast_{\hbar, S}\widehat
g)(\zeta)= \int_{\oR^{2d}} \widehat f(\zeta')\widehat
g(\zeta-\zeta')e^{-(i\hbar/2)\zeta'\cdot (J+ S)(\zeta-\zeta')}d\zeta'.
 \label{6.14}
\end{equation}
(multiplied by $(2\pi)^{-d}$). It  follows from the definition~\eqref{6.11}  and relation~\eqref{6.12} that
\begin{equation}
\Op_{S}(f\star_{\hbar,S} g)=\Op_{S}(f)\Op_{S}(g),
 \notag
\end{equation}
 as can  be verified by using the symmetry of the matrix $S$ and the antisymmetry of $J$, see, e.g.,  Sect.~4 in~\cite{S2019} for more details.
   Let $j_S$ be the operator of multiplication by $e^{(i\hbar/4)\zeta\cdot S\zeta}$. It is easy to see that the deformed convolution~\eqref{6.14} and the twisted convolution~\eqref{5.1} are connected by the relation
\begin{equation}
 j_S(g_1\ast_{\hbar, S} g_2)= j_S(g_1)\ast_\hbar j_S(g_2).
 \notag
\end{equation}
As already noted above, Theorem~1 in~\cite{S2019} shows that under the condition  $b_n\subset a_n$, the function   $e^{(i\hbar/4)\zeta\cdot S\zeta}$ is a pointwise multiplier for $\Sba$ and for $\sba$. Hence these spaces are algebras under the deformed convolution~\eqref{6.14}  and the map  $g(\zeta)\to e^{(i\hbar/4)\zeta\cdot  S\zeta} g(\zeta)$ is an algebraic and topological isomorphism of $\bigl(\Sba,\ast_{\hbar, S}\bigr)$ onto $\bigl(\Sba,\ast_\hbar\bigr)$ and of $\bigl(\sba,\ast_{\hbar, S}\bigr)$ onto $\bigl(\sba,\ast_\hbar\bigr)$. The spaces of multipliers with respect to the products  $\ast_{\hbar, S}$ and $\star_{\hbar, S}$ are defined in complete analogy with~\eqref{5.4}--\eqref{5.7}, \eqref{5.13}, \eqref{5.14}, and we obtain another corollary of Theorem~\ref{T6.2}.

\begin{corollary}\label{C6.5}  Let  $b_n\subset a_n$,  let $S$ be a real symmetric matrix, and let ${\mM}^S_{\hbar}\bigl(\Sab\bigr)$ and ${\mM}^S_{\hbar}\bigl(\sab\bigr)$ be, respectively, the algebras of two-sided multipliers for  $\bigl(\Sab, \star_{\hbar, S}\bigr)$  and  $\bigl(\sab, \star_{\hbar, S}\bigr)$, where $\star_{\hbar, S}$ is defined by~\eqref{6.13}. Then we have  the  continuous embeddings
\begin{equation}
 \eab\hookrightarrow {\mM}^S_\hbar\bigl(\Sab\bigr),\qquad \eabp\hookrightarrow {\mM}^S_\hbar\bigl(\sab\bigr).
\label{6.15}
 \end{equation}
 \end{corollary}
 \begin{proof} Since $\mathcal C^S_\hbar\bigl(\Sba\bigr)$ and $\mathcal C^S_\hbar\bigl(\sba\bigr)$ are obtained from the algebras~\eqref{5.14} by multi\-plication by $e^{-(i\hbar/4)\zeta\cdot  S\zeta}$ and the spaces $M'\bigl(\Sba\bigr)$ and $M'\bigl(\sba\bigr)$ are invariant under this operation,
 it follows from~\eqref{6.8} that
 \begin{equation}
M'\bigl(\Sba\bigr) \hookrightarrow  \mathcal C^S_\hbar\bigl(\Sba\bigr) ,\qquad M'\bigl(\sba\bigr) \hookrightarrow   \mathcal C^S_\hbar\bigl(\sba\bigr).
\notag
\end{equation}
After Fourier transforming and applying Corollary~\ref{C4.5}, we obtain~\eqref{6.15}.
\end{proof}

\section{Concluding remarks}
\label{S7}
The quantization map~\eqref{6.11} extends uniquely to a continuous bijection from $S_{\{a\}}^{\{a\}\prime}(\oR^{2d})$ onto the space $\mathcal L\bigl(S_{\{a\}}^{\{a\}}(\oR^d), S_{\{a\}}^{\{a\}\prime}(\oR^d)\bigr)$, as well as to a continuous bijection from  $S_{(a)}^{(a)\prime}(\oR^{2d})$ onto  $\mathcal L\bigl(S_{(a)}^{(a)}(\oR^d),S_{(a)}^{(a)\prime}(\oR^d)\bigr)$ (for a proof, see Theorem~2 in~\cite{S2012-II}). These extensions are analogous to the extension of the Weyl map to tempered distributions in~\cite{Fol,H3}.
 In this way  $\Op_S(u)$ is well defined for any $u\in S_{\{a\}}^{\{a\}\prime}(\oR^{2d})$ as a continuous linear map of $S_{\{a\}}^{\{a\}}(\oR^d)$ into  $S_{\{a\}}^{\{a\}\prime}(\oR^d)$ and coincides with the operator whose Weyl symbol is $F^{-1}\bigl[\hat u e^{-(i\hbar/4)\zeta\cdot S\zeta}\bigr]$. It follows directly from the definitions 
 that the algebra ${\mM}^S_\hbar \bigl(S_{\{a\}}^{\{a\}}(\oR^{2d})\bigr)$ is transformed by the extended map~\eqref{6.11} into the same set of operators as  ${\mM}_\hbar\bigl(S_{\{a\}}^{\{a\}}(\oR^{2d})\bigr)$  by the Weyl map. In addition, Corollary~\ref{C4.6} implies that the image of  $\breve{\mathcal E}^{\{a\}}_{(a)}(\oR^{2d})$ under the map $u\to \Op_S(u)$ is the same as its image under the Weyl map.
 Analogous statements are true regarding ${\mM}^S_\hbar\bigl(S_{(a)}^{(a)}(\oR^{2d}\bigr)$ and $\breve{\mathcal E}^{(a)}_{\{a\}}(\oR^{2d})$.

Theorem~3 in~\cite{S2012-II} shows that the Weyl map transforms the algebra ${\mM}_{\hbar, L}\left(S^\alpha_\alpha(\oR^{2d})\right)$  of left Moyal multipliers for $S^\alpha_\alpha(\oR^{2d})$ into the algebra of operators mapping  $S^\alpha_\alpha(\oR^{d})$ continuously into itself.
In~\cite{S2020}, this result is extended to the general case of spaces  $S_{(a)}^{(a)}$ and $S_{(a)}^{(a)}$, and it implies, in particular, that the pseudodifferential operators whose Weyl symbols belong to $\breve{\mathcal E}^{\{a\}}_{(a)}$ are continuous on  $S_{\{a\}}^{\{a\}}$ and the operators with Weyl symbols in $\breve{\mathcal E}^{(a)}_{\{a\}}$ are continuous on $S_{(a)}^{(a)}$.
 Pseudodifferential operators with symbols in the spaces $\Gamma^\infty_s$ and $\Gamma^\infty_{0,s}$, which coincide with   $\breve{\mathcal E}^{\{a\}}_{(a)}$ and  $\breve{\mathcal E}^{(a)}_{\{a\}}$ for $a_n=n^{sn}$, were studied by Cappiello and Toft in~\cite{CT}. Their continuity properties are proved there  by a different method based on using modulation spaces and the short time Fourier transform. Similar results on the continuity properties of pseudodifferential operators in the Gelfand-Shilov setting were also derived in another way by Prangoski~\cite{P}, but for slightly smaller symbol classes than   $\breve{\mathcal E}^{\{a\}}_{(a)}$ and $\breve{\mathcal E}^{(a)}_{\{a\}}$. It follows from the above that the continuity properties of operators obtained  from $\breve{\mathcal E}^{\{a\}}_{(a)}$ and $\breve{\mathcal E}^{(a)}_{\{a\}}$ by  applying the map~\eqref{6.11}  with $S\ne 0$ are the same as operators obtained by applying the Weyl map and do not require a separate consideration.

Besides  the spaces~\eqref{2.15} and~\eqref{2.16}, Palamodov introduced in~\cite{P1962} two more classes of spaces, which in our notation are defined by $\mathcal E^{(b)}_{(a)}=\bigcap_{A\to\infty,B\to0}\mathcal E^{b,B}_{a,A}$ and $\mathcal E^{\{b\}}_{\{a\}}=\bigcup_{A\to0,B\to\infty}\mathcal E^{b,B}_{a,A}$. The dual of $\mathcal E^{(b)}_{(a)}$ is the space of convolutors for $S^{\{a\}}_{(b)}=\bigcap_{B\to0}\bigcup_{A\to\infty}S^{b,B}_{a,A}$ and the dual of $\mathcal E^{\{b\}}_{\{a\}}$ is the space of convolutors for $S^{\{b\}}_{(a)}=\bigcap_{A\to0}\bigcup_{B\to\infty}S^{b,B}_{a,A}$. We note that the symbol class  denoted in~\cite{CT} by $\Gamma^\infty_{1,s}$ is, in our notation,  $\mathcal E^{\{a\}}_{\{a\}}$ for $a_n=n^{sn}$. It is clear that $\mathcal E^{(b)}_{(a)}\subset \eab\bigcap\eabp$.
  Therefore, the pseudodifferential operators with symbols in  $\mathcal E^{(a)}_{(a)}(\oR^{2d})$ are continuous from $S^{\{a\}}_{\{a\}}(\oR^d)$ to $S^{\{a\}}_{\{a\}}(\oR^d)$ and from $S^{(a)}_{(a)}(\oR^d)$ to $S^{(a)}_{(a)}(\oR^d)$. In the case of symbols in $\mathcal E^{\{a\}}_{\{a\}}(\oR^{2d})$, the situation  is completely different because $S_{(a)}^{\{a\}}(\oR^{2d})$ is not closed under the Weyl-Moyal product. Since $\mathcal E^{\{a\}}_{\{a\}}\subset S^{(a)\prime}_{(a)}$, its corresponding operators  are well defined as elements of  $\mathcal L\bigl(S_{(a)}^{(a)}(\oR^d),S_{(a)}^{(a)\prime}(\oR^d)\bigr)$, but  $\mathcal E^{\{a\}}_{\{a\}}(\oR^{2d})$ includes functions $u$ such that the image of $S_{(a)}^{(a)}(\oR^d)$ under $\Op(u)$ is not contained in  $L^2(\oR^d)$.

\end{document}